\documentclass[journal]{IEEEtran}
\usepackage{cite}
\usepackage{amsmath,amssymb,amsfonts}
\usepackage{amsthm}
\usepackage{algorithmic}
\usepackage{graphicx}
\usepackage{textcomp}
\usepackage{xcolor}
\usepackage{cuted}
\usepackage{epstopdf}
\usepackage{stfloats}
\usepackage{algorithm}
\usepackage{enumerate}
\usepackage{comment}
\usepackage[caption=false,font=small,labelfont=sf,textfont=sf]{subfig}

\theoremstyle{plain}

\newtheorem{lemma}{Lemma}
\newtheorem{proposition}{Proposition}
\newtheorem{corollary}{Corollary}
\newtheorem{remark}{Remark}

\newcommand{\tr}{\mathrm{tr}}
\newcommand{\diag}{\mathrm{diag}}

\newcommand{\Rea}{\mathrm{Re}}

\begin{document}
 

\title{Two-Timescale Design for Rotatable-Antenna Systems With Imperfect CSI: Rate Analysis \\and Orientation Optimization}

\author{Ziyuan~Zheng, Qingqing~Wu, and Wen Chen
\vspace{-30pt}
\thanks{Z. Zheng, Q. Wu, and W. Chen are with the School of Integrated Circuits, Shanghai Jiao Tong University, 200240, China (e-mail: \{zhengziyuan2024, qingqingwu, wenchen\}@sjtu.edu.cn).}
}

\markboth{}%
{Shell \MakeLowercase{\textit{et al.}}: Bare Demo of IEEEtran.cls for IEEE Journals}

\maketitle

\begin{abstract}
This paper studies uplink multiuser MIMO with a rotatable antenna (RA) array under imperfect channel state information (CSI), where each base-station antenna can adjust its boresight direction within an angular region. 
To balance performance and control overhead, we propose a two-timescale design: RA orientations are optimized from statistical CSI on a large timescale, while linear receive combiners are updated per coherence block from linear minimum-mean-squared-error (LMMSE) channel estimates.
Under this framework, we derive a closed-form use-and-then-forget (UatF)-based rate expression for maximum-ratio combining (MRC) and a closed-form statistical rate surrogate for weighted zero-forcing (wZF) under imperfect CSI, revealing how RA rotation influences useful signal strength, estimation-error-induced self-interference, and multiuser interference. 
The analysis shows that the orientation minimizing channel-estimation error differs from the rate-maximizing one, and that MRC and wZF prefer different rotation configurations due to their distinct mechanisms of signal aggregation and error-aware user separation.
For the resulting non-convex rotation design problems, we develop a projected-gradient algorithm over a product of spherical caps with explicit derivatives of the required channel statistics and rate metrics.
Numerical results verify the accuracy of the large-timescale surrogates and show substantial performance gains from RA optimization.
\end{abstract}


\IEEEpeerreviewmaketitle

\vspace{-12pt}
\section{Introduction}

Multiple-input multiple-output (MIMO) is a key technology for improving spectral efficiency by exploiting spatial degrees of freedom (DoF) \cite{ref1,ref2}. As wireless systems evolve toward larger bandwidths and higher frequencies, large antenna arrays become increasingly important for compensating propagation loss and spatially multiplexing multiple users \cite{ref3,ref4}. Yet in many practically relevant regimes, merely increasing the number of antennas or the transmit power often yields diminishing returns. When user channels are highly correlated, the line-of-sight (LoS) component is dominant, or the channel is only imperfectly known, conventional fixed-pattern arrays leave a fraction of the available DoF unused \cite{ref5}. This limitation has motivated a broader shift from \emph{signal-only adaptation} to \emph{physical-layer adaptation}, in which one may improve a wireless link in three ways: reshape the propagation environment, reshape the array geometry, or reshape the antenna directivity \cite{ref6}. The first direction is represented by reconfigurable radio environments such as intelligent reflecting surfaces (IRSs) \cite{ref7,ref8}. The second is represented by movable or fluid antennas, which adapt antenna positions to exploit favorable spatial points \cite{ref9,ref10,ref11,ref12,ref13}. The third direction keeps antenna locations fixed but changes the directional response of each array element by rotating its boresight, which is appealing when the array layout and radio frequency inter-connection should remain unchanged, but a finer degree of \emph{direction-domain adaptation} is still desired \cite{ref14,ref15,ref16}.

Rotatable antennas (RAs) provide such a mechanism with hardware refinement \cite{ref16,ref17,ref18,ref19}. In an RA-enabled array, each element is directional and can steer its boresight within a feasible angular region. Since the element gain depends on the incident direction, RA rotations reshape both the deterministic LoS mean and the scattering covariance of the wireless channel statistics, affecting both channel state information (CSI) quality and transmission spectral efficiency in uplink reception. In other words, RA rotation acts before digital combining and thus couples physical array control with channel estimation and multiuser detection. This coupling is precisely what makes RA design both promising and challenging. On the one hand, rotating the array can improve channel strength, reduce interference coupling, and sharpen spatial discrimination without moving the antenna positions. On the other hand, the optimal rotation depends on channel statistics that are much slower than those of small-scale fading, whereas the receive combiner should still respond to instantaneous channel realizations. Updating all element orientations once per coherence block is therefore generally impractical \cite{ref20,ref21,ref22}. It would incur substantial mechanical or electronic control overhead and, more importantly, would require instantaneous CSI acquisition over many candidate orientations. This creates a natural timescale mismatch between fast fading and slow orientation control.

A two-timescale design is therefore a natural and practically motivated framework for RA-enabled multiuser reception. On the large timescale, one should optimize the RA orientations using slowly varying statistical CSI determined by user geometry, scatterer geometry, and large-scale fading. On the small timescale, one should update the receive combiner from the instantaneous channel estimates corresponding to the selected orientation. Similar timescale separation has proved highly effective in other physically reconfigurable wireless architectures, such as IRS-assisted systems and statistical-CSI-based movable-antenna designs \cite{ref23,ref24}. However, RA systems are different from both of these lines of work. Unlike IRSs, RAs act directly on the receiving array rather than through a cascaded reflected channel. Unlike movable antennas, RAs do not change propagation distances or array-element locations; instead, they change the element-wise directional sensitivity and, hence, the channel statistics in a different, more structured manner. As a result, existing models, optimization strategies, and insights for IRSs or movable antennas do not directly address the key questions posed by RA-enabled uplink MU-MIMO, especially under imperfect CSI.

This paper addresses these gaps by developing a two-timescale uplink MU-MIMO design for RA systems. The key idea is to treat RA orientation as a large-timescale statistical design variable and to track how it propagates through the geometry-based channel, the linear minimum mean squared error (LMMSE) estimator, and the maximum ratio combining (MRC) and weighted zero-forcing (wZF) receivers. The main contributions of this paper are summarized as follows.
\begin{enumerate}
\item For the RA-enabled uplink MU-MIMO system with geometry-based channel, we propose a two-timescale architecture and formulate rotation optimization problems with imperfect CSI. Specifically, the RA rotation matrix is optimized using statistical CSI on a large timescale, while the receive combiner is updated using instantaneous LMMSE channel estimates on a small timescale.
\item Under the proposed two-timescale framework, we derive a closed-form use-and-then-forget (UatF)-based rate expression for MRC and a closed-form statistical rate surrogate for wZF induced by non-isotropic estimation errors. These expressions reveal how RA rotation affects useful signal strength, estimation-error-induced self-interference, and multiuser interference.
\item Our analysis further reveals that RA rotation can improve CSI quality by strengthening the channel covariance over the active subspace, but the orientation minimizing NMSE generally differs from the one maximizing communication rate, and MRC and wZF prefer different rotation configurations because they emphasize signal aggregation and error-aware user separation differently. 
\item Due to the highly non-convex dependence of the closed-form rate surrogates under MRC and wZF on the RA rotation matrix, we develop a unified algorithm framework based on the projected-gradient method for the RA rotation optimization problems under MRC and wZF over a product of spherical caps. We derive explicit gradients of rate surrogates and derivatives of channel statistics.
\item In the numerical results, the closed-form surrogates closely match block-level ergodic performance, validating the two-timescale design. Optimized RA orientations outperform random or fixed orientations, especially in interference-limited regimes. The theoretical insights are verified, and wZF benefits more strongly from RA rotation because it exploits error-aware user separation.
\end{enumerate}

The remainder of this paper is organized as follows. Section~\ref{sec:system_model} presents the system model. Section~\ref{sec:channel_estimation} presents LMMSE estimation and error statistics. Section~\ref{sec:problem_formulation} formulates the two-timescale optimization problem. Section~\ref{sec:rate_analysis} derives closed-form rate surrogates. Section~\ref{sec:rotation_optimization} develops the projected-gradient algorithm. Section~\ref{sec:numerical_results} presents numerical results, and Section~\ref{sec:conclusion} concludes the paper.

\setlength{\abovecaptionskip}{0pt}
 \begin{figure}
\centering
\includegraphics[width=3.5in]{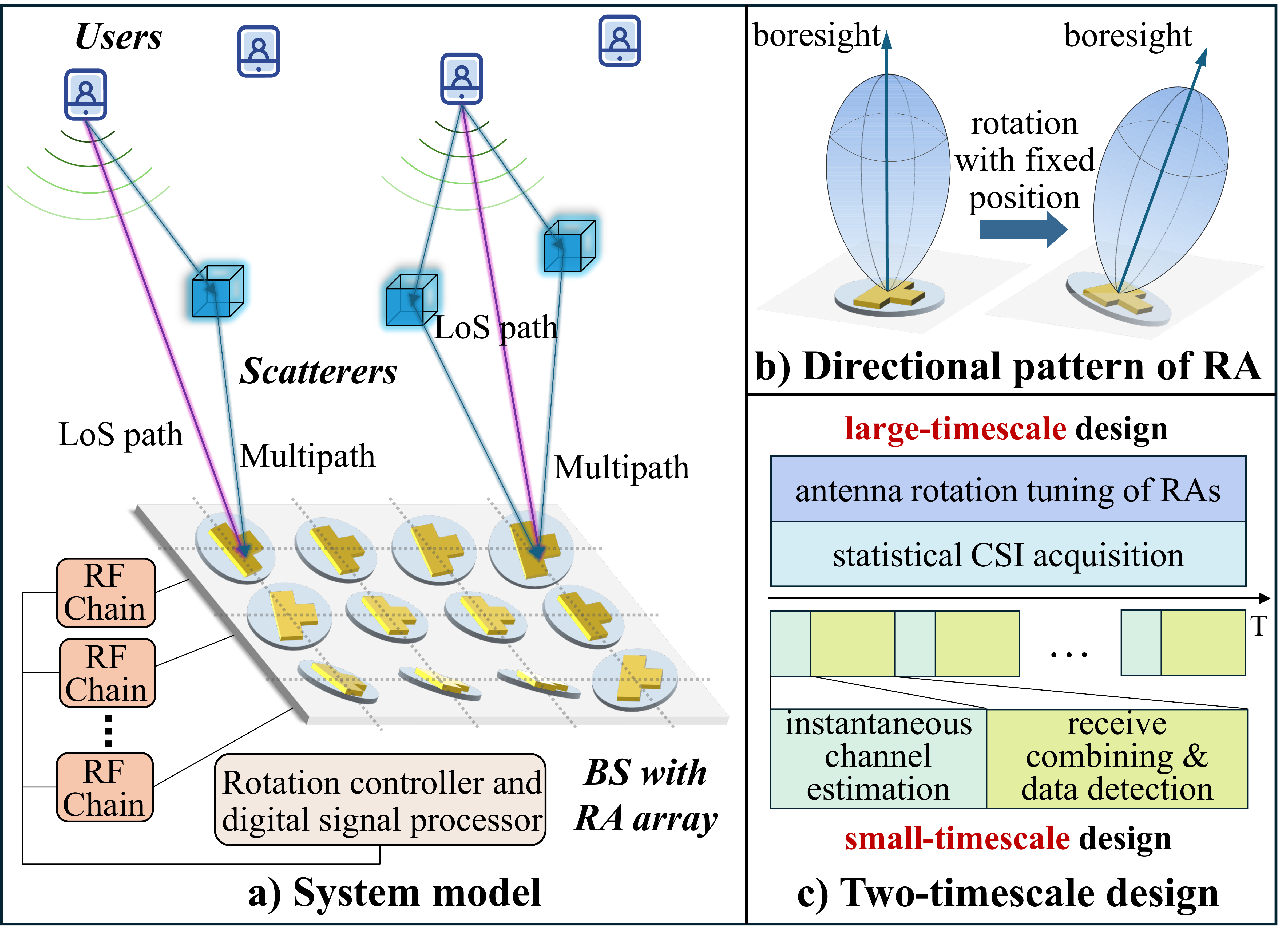}
\captionsetup{font=small}
\caption{Illustration of the system model: a) the RA-enabled uplink multiuser MIMO system, b) directional radiation pattern and rotations of RA, and c) the framework of the proposed two-timescale design.} 
\label{fig:system_model}
\vspace{-15pt}
\end{figure}

\vspace{-9pt}
\section{System Model}\label{sec:system_model}
\vspace{-3pt}

As illustrated in Fig.~\ref{fig:system_model}, we consider a multiuser MIMO system, where a base station (BS) equipped with a uniform planar array (UPA) of $N$ rotatable antennas (RA) simultaneously serves $K$ single-antenna users for uplink communication. Let $\mathcal{N} \triangleq \{1,\dots,N\}$ and $\mathcal{K} \triangleq \{1,\dots,K\}$ denote the index sets of RA elements and users, respectively. 

\vspace{-5pt}
\subsection{Rotation Vectors and Directional Gain of RA}
\vspace{-2pt}

We adopt a three-dimensional (3D) Cartesian coordinate system, in which the RA array lies on the $x$–$O$-$y$ plane, with its center located at the origin $O$. The position of RA $n$ is denoted by $\boldsymbol{w}_n \triangleq \left[x_n,  y_n,  0\right]^T \in \mathbb{R}^{3 \times 1}$, where $(x_n,y_n)$ are determined by the UPA geometry with half-wavelength inter-element spacing and arrangement of $N=N_\text{col} \times N_\text{row}$ elements. The position of user $k$ is denoted by $\boldsymbol{u}_k \triangleq [x_k,  y_k,  z_k]^T \in \mathbb{R}^{3 \times 1}$. Accordingly, we denote $r_{k,n}$ and $\boldsymbol{s}_{k,n}$ as the distance between user $k$ and RA $n$ and the unit direction vector from RA $n$ to user $k$, respectively, which can be given by
\begin{align}
r_{k,n} \triangleq \|\boldsymbol{u}_k - \boldsymbol{w}_n\|,\,\,
\boldsymbol{s}_{k,n} \triangleq \frac{\boldsymbol{u}_k-\boldsymbol{w}_n}{\lVert \boldsymbol{u}_k-\boldsymbol{w}_n \rVert}=\frac{\boldsymbol{u}_k-\boldsymbol{w}_n}{r_{k,n}}.
\label{eq:incident_signal_distance_direction}
\end{align}

Each RA is implemented as a directional antenna whose boresight direction can be independently adjusted by mounting it on a rotation platform, while its physical position $\boldsymbol{w}_n$ remains fixed. The boresight direction of RA $n$ is characterized by a unit pointing vector $\boldsymbol{f}_n\triangleq \left[ f_{x,n},f_{y,n},f_{z,n} \right] ^T\in \mathbb{R}^{3\times 1}$ with $\lVert \boldsymbol{f}_n \rVert =1,\forall n\in \mathcal{N}$, where $f_{x,n}$, $f_{y,n}$, and $f_{z,n}$ denote the projections of the pointing vector on the $x$-, $y$-, and $z$-axes, respectively. 
Equivalently, $\boldsymbol{f}_n$ can be parameterized by an elevation angle $\theta_n\in[0,\pi]$ between the boresight direction and the positive $z$-axis and an azimuth angle $\phi_n\in[-\pi,\pi)$ between the projection of the boresight direction onto the $x$-$O$–$y$ plane and the positive $x$-axis as $\boldsymbol{f}_n=\left[ \sin \theta _n\cos \phi _n, \sin \theta _n\sin \phi _n,\cos \theta _n \right]^T$.
We collect the pointing vectors of all RAs into the matrix 
$\boldsymbol{F} \triangleq [ \boldsymbol{f}_1,\dots,\boldsymbol{f}_N ] \in \mathbb{R}^{3 \times N}$,
which serves as the rotation configuration of the RA array.

In practice, the rotation range of each RA is limited within a finite angular region due to mechanical constraints. Let $\boldsymbol{e}_z\triangleq[0,0,1]^{T}$ denote the array normal direction. By imposing a constraint related to feasible orientation $\arccos(\boldsymbol{f}_n^{T}\boldsymbol{e}_z)$, the feasible set of boresight vectors for each RA is thus
\begin{align}
\mathcal{F}_n \triangleq \big\{\boldsymbol{f}_n\in\mathbb{R}^3:\|\boldsymbol{f}_n\|=1,~\arccos(\boldsymbol{f}^{T}_n\boldsymbol{e}_z)\le \theta_{\max}\big\},
\label{eq:feasible_region}
\end{align}
where $\theta_{\max}\in(0,\pi/2]$ specifies the maximum allowable tilt angle from the array normal. Collectively, the feasible set of orientation vectors is $\mathcal{F}$. The radiation pattern of each RA depends on the incident signal direction $\boldsymbol{s}_{k,n}$ in \eqref{eq:incident_signal_distance_direction}. Define the mismatch angle between $\boldsymbol{f}_n$ and $\boldsymbol{s}_{k,n}$ as 
\begin{align}
\vartheta_{k,n} \triangleq \arccos ( \boldsymbol{f}_n^T \boldsymbol{s}_{k,n}), \forall k \in \mathcal{K},\forall n \in \mathcal{N}.
\end{align}
Following a cosine-based parametric model of half-space directional gain, the normalized element pattern of RAs is
\begin{align}
G\left( \boldsymbol{f}_{n},\boldsymbol{s}_{k,n} \right) = G_0\,[ \boldsymbol{f}_{n}^{T}\boldsymbol{s}_{k,n} ]_+^{2b},
\label{eq:element_gain}
\end{align}
where $[\cdot]_+\triangleq \max\{\cdot,0\}$ guarantees a non-negative value, $b\ge 0$ is the directional factor determining the main-lobe beamwidth, and $G_0 \triangleq 2(2b+1)$ is the maximum gain satisfying the power normalization and achieved when the incident wave is aligned with the RA boresight, i.e., $\boldsymbol{f}_{n}^{T}\boldsymbol{s}_{k,n}=1$.

\vspace{-9pt}
\subsection{Geometry-Based Rician Channel Model}
\vspace{-3pt}

We adopt a geometry-based narrow-band Rician fading model for the RA-enabled uplink channels. The propagation environment is characterized by a dominant line-of-sight (LoS) path and $Q$ non-LoS (NLoS) scatterer clusters distributed in space. 
For the LoS path between user $k$ and RA $n$, the large-scale path gain is modeled as
\begin{align}
g_{k,n}^{\mathrm{LoS}}(\boldsymbol{f}_n)
=\varrho \frac{G(\boldsymbol{f}_{n},\boldsymbol{s}_{k,n})}{4\pi r_{k,n}^2},
\label{eq:LoS_gain}
\end{align}
where $\varrho$ is the physical size of each RA. For the NLoS component, the signal from user $k$ is first scattered by cluster $q\in\mathcal{Q}$, then received by RA $n$. The position of scatterer cluster $q$ is denoted by $\boldsymbol{c}_q\in\mathbb{R}^3$ and its corresponding set is $\mathcal{Q}\triangleq\{1,\ldots,Q\}$. Following \eqref{eq:incident_signal_distance_direction}, the distance and unit direction vector from RA $n$ to cluster $q$ are denoted by $r_{q,n} \triangleq \|\boldsymbol{c}_q - \boldsymbol{w}_n\|$ and $\boldsymbol{s}_{q,n} \triangleq \frac{\boldsymbol{c}_q-\boldsymbol{w}_n}{\lVert \boldsymbol{c}_q-\boldsymbol{w}_n \rVert}=\frac{\boldsymbol{c}_q-\boldsymbol{w}_n}{r_{q,n}}$, respectively. Then, similar to \eqref{eq:LoS_gain}, the directional path gain associated with cluster $q$ at RA $n$ is modeled as
\begin{align}
g_{q,n}^{\mathrm{c}}(\boldsymbol{f}_n)
= \varrho\frac{ G( \boldsymbol{f}_n,\boldsymbol{s}_{q,n}) }{4\pi r_{q,n}^2}. 
\label{eq:cluster_gain}
\end{align}
We further define $d_{q,k}\triangleq \|\boldsymbol{u}_k-\boldsymbol{c}_q\|$ the distance from cluster $q$ to user $k$.
Next, we take small-scale fading into account for the complex baseband channel coefficient between user $k$ and RA $n$ with the multipath
\begin{align}
h_{k,n}(\boldsymbol{f}_n)
= h_{k,n}^{\mathrm{LoS}}(\boldsymbol{f}_n) + h_{k,n}^{\mathrm{NLoS}}(\boldsymbol{f}_n). 
\label{eq:hk_n_decomp}
\end{align}
For the LoS component, we have
\begin{align}
h_{k,n}^{\mathrm{LoS}}(\boldsymbol{f}_n)
= \sqrt{g_{k,n}^{\mathrm{LoS}}(\boldsymbol{f}_n)} 
e^{-j\frac{2\pi}{\lambda} r_{k,n}},
\label{eq:hk_n_LoS}
\end{align}
where $\lambda$ denotes the carrier wavelength. For the NLoS component, we model the contribution of each cluster as an independent complex    
path and obtain
\begin{align}
h_{k,n}^{\text{NLoS}}\left( \boldsymbol{f}_n \right) =\sum_{q=1}^Q{\varpi _{k,q}\frac{\sqrt{\sigma _qg_{q,n}^{\text{c}}\left( \boldsymbol{f}_n \right)}}{d_{q,k}}}e^{-j\frac{2\pi}{\lambda}\left( r_{q,n}+d_{q,k} \right)},
\label{eq:hk_n_NLoS}
\end{align}
where $\sigma_q>0$ represents the radar cross section of scatterer cluster $q$, the random variables $\varpi_{k,q}\sim\mathcal{CN}(0,1)$ model the user–cluster small-scale fading. The variables $\{\varpi_{k,q}\}$ are independent across $(k,q)$, and independent of the LoS component.

For a given RA orientation matrix $\boldsymbol{F}$ and stable geometry $\{\boldsymbol{u}_k\}$, $\{\boldsymbol{w}_n\}$ and $\{\boldsymbol{c}_q\}$, the randomness of $\boldsymbol{h}_k(\boldsymbol{F})$ mainly stems from the small-scale fading coefficients $\{\varpi_{k,q}\}$ The uplink channel vector $\boldsymbol{h}_k(\boldsymbol{F})$ from user $k$ can be written in the general Rician form
\vspace{-3pt}
\begin{align}
\boldsymbol{h}_k(\boldsymbol{F})
= \boldsymbol{\mu}_k(\boldsymbol{F}) + \boldsymbol{B}_k(\boldsymbol{F}) \boldsymbol{\varpi}_k, 
\label{eq:hk_expression}
\end{align}
where $\boldsymbol{\mu}_k(\boldsymbol{F})\in\mathbb{C}^{N\times 1}$ collects the deterministic LoS components across the RAs, $\boldsymbol{B}_k (\boldsymbol{F}) \in \mathbb{C}^{N\times Q}$ is a deterministic matrix determined by the NLoS clusters and the RA orientation-dependent gains $g_{q,n}^{\mathrm{c}}(\boldsymbol{f}_n)$, and $\boldsymbol{\varpi}_k\sim\mathcal{CN}(\boldsymbol{0},\boldsymbol{I}_{Q})$ stacks the independent and identically distributed (i.i.d.) Gaussian path gains associated with these clusters. Consequently, we have
\begin{align}
&\boldsymbol{h}_k(\boldsymbol{F}) \sim \mathcal{CN}\big(\boldsymbol{\mu}_k(\boldsymbol{F}), \boldsymbol{R}_k(\boldsymbol{F})\big),
\label{eq:hk_distribution}
\\
&\boldsymbol{R}_k(\boldsymbol{F}) \triangleq \boldsymbol{B}_k(\boldsymbol{F})\boldsymbol{B}_k^H(\boldsymbol{F}).
\end{align}
The pair $\big(\boldsymbol{\mu}_k(\boldsymbol{F}),\boldsymbol{R}_k(\boldsymbol{F})\big)$ characterizes the Rician channel statistics of user $k$ under a given $\boldsymbol{F}$, playing a central role in the subsequent channel estimation and achievable rate analysis.

\vspace{-9pt}
\subsection{Uplink Signal Model}
\vspace{-1pt}

The uplink received signal at the RA array is written as
\begin{align}
    \boldsymbol{y} = \sum_{k\in \mathcal{K}} \sqrt{p_k} \boldsymbol{h}_k(\boldsymbol{F}) x_k + \boldsymbol{n},
    \label{eq:uplink_signal}
\end{align}
where $p_k$ is the transmit power of user $k$ subject to $0\le p_k \le P_k^{\max}$, $x_k$ is the normalized information symbol transmitted by user $k$ with $\mathbb{E}\{|x_k|^2\}=1$ and $\boldsymbol{x} \triangleq [x_1,\dots,x_K]^T$, and $\boldsymbol{n} \sim \mathcal{CN}(\boldsymbol{0},\sigma^2 \boldsymbol{I}_N)$ is the additive white Gaussian noise (AWGN) at the BS.
In each coherence block with a rotation configuration $\boldsymbol{F}$, the BS first performs uplink channel estimation based on pilot transmissions to obtain an estimate $\hat{\boldsymbol{H}}(\boldsymbol{F})$ of $\boldsymbol{H}(\boldsymbol{F})$, where $\boldsymbol{H}(\boldsymbol{F}) \triangleq [\boldsymbol{h}_1(\boldsymbol{F}),\dots,\boldsymbol{h}_K(\boldsymbol{F})] \in \mathbb{C}^{N \times K}$ is the channel matrix stacking channels of all users, and then constructs a linear receive combining matrix $\boldsymbol{V} \triangleq [ \boldsymbol{v}_1,\dots,\boldsymbol{v}_K ]
 \in \mathbb{C}^{N \times K}$, so that the detected symbol for user $k$ is given by
\begin{equation}
    \varsigma_k = \boldsymbol{v}_k^H \boldsymbol{y}, \quad \forall k \in \mathcal{K}.
    \label{eq:rx_combining}
\end{equation}
Different choices of $\boldsymbol{V}$ lead to different linear receivers, e.g., MRC or wZF, in Section~\ref{sec:rate_analysis}. 
This operation motivates a two-timescale design: in each coherence block, the BS updates the receive beamforming matrix $\boldsymbol{V}$ based on the instantaneous channel estimates, while the RA orientations $\boldsymbol{F}$ are updated based on slowly varying channel statistics affecting $\big\{\boldsymbol{\mu}_k(\boldsymbol{F}),\boldsymbol{R}_k(\boldsymbol{F})\big\}_{k\in\mathcal{K}}$ determined by propagation geometries. 

Before designing the large-timescale RA orientations, it is necessary to characterize the small-timescale CSI available at the BS for a fixed orientation matrix \(F\). In particular, since \(F\) reshapes both the Rician mean and covariance, it also affects the LMMSE channel estimate and the associated estimation-error statistics. The next section derives these orientation-dependent estimation statistics, which will later enter the MRC and wZF rate expressions, as well as the RA rotation design.

\vspace{-6pt}
\section{Channel Estimation}\label{sec:channel_estimation}
\vspace{-2pt}
In this section, we consider a time-division duplexing (TDD) frame structure, where each coherence block of length $T_{\mathrm{c}}$ symbol intervals is partitioned into a pilot phase of length $\tau_{\mathrm{p}}$ symbols and a data transmission phase of length $T_{\mathrm{d}}=T_{\mathrm{c}}-\tau_{\mathrm{p}}$. Within each coherence block, $\boldsymbol{F}$ is fixed and determined by statistical CSI, and the geometry $\{\boldsymbol{u}_k\}$, $\{\boldsymbol{w}_n\}$ and $\{\boldsymbol{c}_q\}$ of channel $\boldsymbol{H}(\boldsymbol{F})$ remains constant, while the fast-varying fading coefficients $\{\varpi_{k,q}\}$ required to be estimated. 

\vspace{-6pt}
\subsection{Uplink Training Protocol}
\vspace{-1pt}

During the pilot phase, each user transmits a predefined pilot sequence $\boldsymbol{\varphi}_k\in\mathbb{C}^{\tau_{\mathrm{p}}\times 1}$ with power $p_k^{\mathrm{tr}}$. The pilot sequences, collected by the matrix
$\boldsymbol{\varPhi}\triangleq\big[\boldsymbol{\varphi}_1,\ldots,\boldsymbol{\varphi}_K\big]\in\mathbb{C}^{\tau_{\mathrm{p}}\times K}$, are mutually orthogonal between different users, satisfying $\boldsymbol{\varPhi}^H \boldsymbol{\varPhi} = \tau_\mathrm{p} \boldsymbol{I}_K$ with $\tau_{\mathrm{p}}\geq K$.
Stacking the received pilot signals over $\tau_\mathrm{p}$ symbols at the $N$-element RA array yields the received pilot signal matrix \cite{ref25,ref26}
\begin{align}
\boldsymbol{Y}_{\mathrm{p}}(\boldsymbol{F})
&= \sum_{k=1}^K \sqrt{p_k^{\mathrm{tr}}} \boldsymbol{h}_k(\boldsymbol{F}) \boldsymbol{\varphi}_k^{T}
+ \boldsymbol{Z}_{\mathrm{p}}, \label{eq:Yp}
\end{align}
where $\boldsymbol{Z}_{\mathrm{p}}\in\mathbb{C}^{N\times\tau_\mathrm{p}}$ collects the AWGN during the pilot phase as a matrix whose entries are i.i.d. as $\mathcal{CN}(0,\sigma^2)$.
To obtain a sufficient statistic for estimating the channel of user $k$, the BS correlates the received pilot matrix $\boldsymbol{Y}_{\mathrm{p}}(\boldsymbol{F})$ with the conjugate pilot sequence $\boldsymbol{\varphi}_k^*$ and exploits the orthogonality by
\begin{align}
\boldsymbol{y}_{\mathrm{p},k}(\boldsymbol{F}) = \frac{1}{\sqrt{\tau_\mathrm{p}}}\boldsymbol{Y}_{\mathrm{p}}(\boldsymbol{F})\boldsymbol{\varphi}_k^*= \sqrt{\tau_\mathrm{p} p_k^{\mathrm{tr}}} \boldsymbol{h}_k(\boldsymbol{F}) + \boldsymbol{n}_{\mathrm{p},k},
\label{eq:ypk}
\end{align}
with $\boldsymbol{n}_{\mathrm{p},k}
\triangleq \frac{1}{\sqrt{\tau_\mathrm{p}}}\boldsymbol{Z}_{\mathrm{p}}\boldsymbol{\varphi}_k^*
\sim \mathcal{CN}(\boldsymbol{0},\sigma^2\boldsymbol{I}_N)$ being independent of $\boldsymbol{h}_k(\boldsymbol{F})$. Hence, for each user $k$, the effective pilot observation is an $N$-dimensional noisy linear observation of the channel vector $\boldsymbol{h}_k(\boldsymbol{F})$ with an effective pilot energy $\tau \mathrm{p}_k^{\mathrm{tr}}$.

\vspace{-6pt}
\subsection{LMMSE Channel Estimator}
\vspace{-3pt}

Given $\boldsymbol{F}$, the prior distribution of the channel vector is \eqref{eq:hk_distribution}. Then, we can characterize the joint distribution of $\boldsymbol{h}_k(\boldsymbol{F})$ and $\boldsymbol{y}_{\mathrm{p},k}(\boldsymbol{F})$. Denote the statistical expectation by $\mathbb{E}[\cdot]$. The mean of $\boldsymbol{y}_{\mathrm{p},k}(\boldsymbol{F})$ is
\begin{align}
\bar{\boldsymbol{y}}_{\mathrm{p},k}(\boldsymbol{F})
&\triangleq \mathbb{E}\big[\boldsymbol{y}_{\mathrm{p},k}(\boldsymbol{F})\big]
= \sqrt{\tau p_k^{\mathrm{tr}}} \boldsymbol{\mu}_k(\boldsymbol{F}),
\label{eq:mean_ypk}
\end{align}
and its covariance matrix is
\begin{align}
\boldsymbol{C}_{\boldsymbol{y}_k}(\boldsymbol{F})
&\triangleq \mathbb{E}\left[(\boldsymbol{y}_{\mathrm{p},k}(\boldsymbol{F})-\bar{\boldsymbol{y}}_{\mathrm{p},k}(\boldsymbol{F}))
(\boldsymbol{y}_{\mathrm{p},k}(\boldsymbol{F})-\bar{\boldsymbol{y}}_{\mathrm{p},k}(\boldsymbol{F}))^{H}\right]\nonumber
\\
&= \tau p_k^{\mathrm{tr}}\boldsymbol{R}_k(\boldsymbol{F}) + \sigma^2\boldsymbol{I}_N. \label{eq:Cov_ypk}
\end{align}
The cross-covariance between $\boldsymbol{h}_k(\boldsymbol{F})$ and $\boldsymbol{y}_{\mathrm{p},k}(\boldsymbol{F})$ is 
\begin{align}
\boldsymbol{C}_{\boldsymbol{h}_k\boldsymbol{y}_k}(\boldsymbol{F})
&\triangleq \mathbb{E}\left[(\boldsymbol{h}_k(\boldsymbol{F})-\boldsymbol{\mu}_k(\boldsymbol{F}))
(\boldsymbol{y}_{\mathrm{p},k}(\boldsymbol{F})-\bar{\boldsymbol{y}}_{\mathrm{p},k}(\boldsymbol{F}))^{H}\right]\nonumber
\\
&= \sqrt{\tau p_k^{\mathrm{tr}}} \boldsymbol{R}_k(\boldsymbol{F}). \label{eq:Cov_hk_ypk}
\end{align}
Since the joint distribution of $\boldsymbol{h}_k(\boldsymbol{F})$ and $\boldsymbol{y}_{\mathrm{p},k}(\boldsymbol{F})$ is complex Gaussian, the LMMSE estimator coincides with the MMSE estimator. Using the standard linear MMSE, the channel estimate of user $k$, given the observation $\boldsymbol{y}_{\mathrm{p},k}(\boldsymbol{F})$, is
\begin{align}
\hat{\boldsymbol{h}}_k(\boldsymbol{F})
&\triangleq \mathbb{E}[\boldsymbol{h}_k(\boldsymbol{F})\mid \boldsymbol{y}_{\mathrm{p},k}(\boldsymbol{F})]\nonumber
\\
&= \boldsymbol{\mu}_k(\boldsymbol{F})
+ \boldsymbol{C}_{\boldsymbol{h}_k\boldsymbol{y}_k}(\boldsymbol{F}) 
\boldsymbol{C}_{\boldsymbol{y}_k}^{-1}(\boldsymbol{F}) 
\big(\boldsymbol{y}_{\mathrm{p},k}(\boldsymbol{F})-\bar{\boldsymbol{y}}_{\mathrm{p},k}(\boldsymbol{F})\big)\nonumber
\\
&= \boldsymbol{\mu}_k(\boldsymbol{F}) 
+ \sqrt{\tau p_k^{\mathrm{tr}}} \boldsymbol{R}_k(\boldsymbol{F})\big(\tau p_k^{\mathrm{tr}}\boldsymbol{R}_k(\boldsymbol{F})+\sigma^2\boldsymbol{I}_N\big)^{-1}\nonumber
\\
&\quad \times \big(\boldsymbol{y}_{\mathrm{p},k}(\boldsymbol{F})-\sqrt{\tau p_k^{\mathrm{tr}}} \boldsymbol{\mu}_k(\boldsymbol{F})\big).
\label{eq:LMMSE_estimator}
\end{align}
\eqref{eq:LMMSE_estimator} captures the dependence of the channel estimate $\hat{\boldsymbol{h}}_k(\boldsymbol{F})$ on the RA orientation $\boldsymbol{F}$ via the channel statistics $\boldsymbol{\mu}_k(\boldsymbol{F})$ and $\boldsymbol{R}_k(\boldsymbol{F})$, which are determined by user and cluster geometries as well as large-scale fading parameters with a fixed $\boldsymbol{F}$.

\vspace{-9pt}
\subsection{Estimation Error Statistics}
\vspace{-3pt}
Define the channel estimation error of user $k$ as
\begin{align}
\tilde{\boldsymbol{h}}_k(\boldsymbol{F})
\triangleq \boldsymbol{h}_k(\boldsymbol{F}) - \hat{\boldsymbol{h}}_k(\boldsymbol{F}). 
\label{eq:error_def}
\end{align}
By standard properties of LMMSE estimation with Gaussian priors, $\tilde{\boldsymbol{h}}_k(\boldsymbol{F})$ is independent of $\hat{\boldsymbol{h}}_k(\boldsymbol{F})$ and follows a complex Gaussian distribution
\begin{align}
\tilde{\boldsymbol{h}}_k(\boldsymbol{F})
\sim \mathcal{CN}\big(\boldsymbol{0}, \boldsymbol{C}_{\mathrm{e},k}(\boldsymbol{F})\big),
  \label{eq:tilde_h_distribution}
\end{align}
where the error covariance matrix is
\begin{align}
\!&\boldsymbol{C}_{\mathrm{e},k}(\boldsymbol{F})
\triangleq \mathbb{E}\big[\tilde{\boldsymbol{h}}_k(\boldsymbol{F})\tilde{\boldsymbol{h}}_k^{H}(\boldsymbol{F})\big]\nonumber\\
\!&= \boldsymbol{R}_k(\boldsymbol{F})
- \boldsymbol{C}_{\boldsymbol{h}_k\boldsymbol{y}_k}(\boldsymbol{F}) 
\boldsymbol{C}_{\boldsymbol{y}_k}^{-1}(\boldsymbol{F}) 
\boldsymbol{C}_{\boldsymbol{y}_k\boldsymbol{h}_k}(\boldsymbol{F})\nonumber
\\
\!&= \!\boldsymbol{R}_k(\boldsymbol{F})\!
-\! \tau p_k^{\mathrm{tr}} 
\boldsymbol{R}_k(\boldsymbol{F})
\big(\tau p_k^{\mathrm{tr}}\boldsymbol{R}_k(\boldsymbol{F})\!+\!\sigma^2\boldsymbol{I}_N\big)^{-1}\!
\boldsymbol{R}_k(\boldsymbol{F}). \!\!\! \nonumber
\\
\!&=\boldsymbol{R}_k(\boldsymbol{F})\big(\boldsymbol{I}_N+\tfrac{\tau p_k^{\mathrm{tr}}}{\sigma^2}\boldsymbol{R}_k(\boldsymbol{F})\big)^{-1},
\label{eq:error_covariance}
\end{align}
In addition, the covariance of the channel estimate is given by
\begin{align}
\boldsymbol{C}_{\hat{h}_k}(\boldsymbol{F})
&\triangleq \mathbb{E}\big[(\hat{\boldsymbol{h}}_k(\boldsymbol{F})-\boldsymbol{\mu}_k(\boldsymbol{F}))
(\hat{\boldsymbol{h}}_k(\boldsymbol{F})-\boldsymbol{\mu}_k(\boldsymbol{F}))^{H}\big]\nonumber
\\
&= \boldsymbol{R}_k(\boldsymbol{F}) - \boldsymbol{C}_{\mathrm{e},k}(\boldsymbol{F})\nonumber
\\
&= \tau p_k^{\mathrm{tr}} 
\boldsymbol{R}_k(\boldsymbol{F})
(\tau p_k^{\mathrm{tr}}\boldsymbol{R}_k(\boldsymbol{F})\!+\!\sigma^2\boldsymbol{I}_N)^{-1}
\boldsymbol{R}_k(\boldsymbol{F}), \!\!\!
\label{eq:est_covariance}
\end{align}
and $\hat{h}_k(\boldsymbol{F})$ follows complex Gaussian distribution
\begin{align}
  \hat{\boldsymbol{h}}_k(\boldsymbol{F})  \sim
  \mathcal{CN}\big(\boldsymbol{\mu}_k(\boldsymbol{F}),
     \boldsymbol{C}_{\hat{h}_k}(\boldsymbol{F})\big),
  \label{eq:hat_h_distribution}
\end{align}

For later evaluation on the estimation error, let $a_k\triangleq \tau p_k^{\mathrm{tr}}$ and for the Hermitian positive semi-definite matrix $\boldsymbol{R}_k(\boldsymbol{F})$,  
denote its Moore-Penrose pseudo-inverse by $\boldsymbol{R}_k^{\dagger}(\boldsymbol{F})$ with the orthogonal projector $\boldsymbol{P}_k(\boldsymbol{F})\triangleq \boldsymbol{R}_k(\boldsymbol{F})\,\boldsymbol{R}_k^{\dagger}(\boldsymbol{F})$ onto the range space $\mathcal{R}(\boldsymbol{R}_k(\boldsymbol{F}))$, which is the subspace spanned by the eigenvectors corresponding to its strictly positive eigenvalues. 
Further define the \emph{normalized error covariance matrix} as
\begin{align}
\boldsymbol{E}_k(\boldsymbol{F})
\triangleq
\boldsymbol{R}_k^{\dagger/2}(\boldsymbol{F})\,\boldsymbol{C}_{e,k}(\boldsymbol{F})\,\boldsymbol{R}_k^{\dagger/2}(\boldsymbol{F}),
\end{align}
where $\boldsymbol{R}_k^{\dagger/2}$ is the Hermitian square root of $\boldsymbol{R}_k^{\dagger}(\boldsymbol{F})$. Accordingly, the normalized mean-square error (NMSE) over the active subspace is defined as
\begin{align}
\mathrm{NMSE}_k(\boldsymbol{F})
\triangleq
\begin{cases}
\frac{1}{\iota_k}\mathrm{tr}\big(\boldsymbol{E}_k(\boldsymbol{F})\big), & \iota_k>0,\\
0, & \iota_k=0,
\end{cases}
\end{align}
with $\iota_k\triangleq \mathrm{rank}\big(\boldsymbol{R}_k(\boldsymbol{F})\big)$.

\begin{lemma}
\label{lem:NMSE}
For any $\boldsymbol{R}_k(\boldsymbol{F})\succeq \boldsymbol{0}$, the normalized error covariance satisfies
\begin{align}
\boldsymbol{E}_k(\boldsymbol{F})
=
\boldsymbol{P}_k(\boldsymbol{F})
\Big(\boldsymbol{I}_N+\tfrac{a_k}{\sigma^2}\boldsymbol{R}_k(\boldsymbol{F})\Big)^{-1}
\boldsymbol{P}_k(\boldsymbol{F}),
\label{eq:NMSE_matrix}
\end{align}
Moreover, if $\{\lambda_{k,i}(\boldsymbol{F})\}_{i=1}^{\iota_k}$ are the non-zero eigenvalues of $\boldsymbol{R}_k(\boldsymbol{F})$ , then the eigenvalues of $\boldsymbol{E}_k(\boldsymbol{F})$ over the active subspace are $\big\{(1+\tfrac{a_k}{\sigma^2}\lambda_{k,i}(\boldsymbol{F}))^{-1}\big\}_{i=1}^{\iota_k}$. Consequently, $\mathrm{NMSE}_k(\boldsymbol{F})$ decreases as the non-zero eigenvalues of $\boldsymbol{R}_k(\boldsymbol{F})$ increase.
In particular, if two covariance matrices with rotations $\boldsymbol{F}_1$ and $\boldsymbol{F}_2$, respectively, satisfy $\boldsymbol{R}_k(\boldsymbol{F}_1)\succeq \boldsymbol{R}_k(\boldsymbol{F}_2)\succeq \boldsymbol{0}$ in the Loewner order and share the same active subspace, i.e., $\mathcal{R}(\boldsymbol{R}_k(\boldsymbol{F}_1))=\mathcal{R}(\boldsymbol{R}_k(\boldsymbol{F}_2))$, then
$\boldsymbol{E}_k(\boldsymbol{F}_1)\preceq \boldsymbol{E}_k(\boldsymbol{F}_2)$ and hence $\mathrm{NMSE}_k(\boldsymbol{F}_1)\le \mathrm{NMSE}_k(\boldsymbol{F}_2)$.
\end{lemma}

\begin{IEEEproof}
For brevity, we omit the explicit dependence on $\boldsymbol{F}$.
Pre-multiplying and post-multiplying by $\boldsymbol{R}_k^{\dagger/2}$ for \eqref{eq:error_covariance} yields
\begin{align*}
\boldsymbol{E}_k
&=\boldsymbol{R}_k^{\dagger/2}\boldsymbol{C}_{e,k}\boldsymbol{R}_k^{\dagger/2}
=\boldsymbol{R}_k^{\dagger/2}\boldsymbol{R}_k\boldsymbol{R}_k^{\dagger/2}
\big(\boldsymbol{I}_N+\tfrac{a_k}{\sigma^2}\boldsymbol{R}_k\big)^{-1}\\
&=\boldsymbol{P}_k
\big(\boldsymbol{I}_N+\tfrac{a_k}{\sigma^2}\boldsymbol{R}_k\big)^{-1}
\boldsymbol{P}_k,
\end{align*}
which proves \eqref{eq:NMSE_matrix}. The eigenvalue characterization follows by applying the spectral decomposition of $\boldsymbol{R}_k$. For each non-zero eigenvalue $\lambda_{k,i}>0$, the corresponding normalized error eigenvalue equals $(1+\frac{a_k}{\sigma^2}\lambda_{k,i})^{-1}$,
and $\boldsymbol{E}_k$ vanishes on the null space $\mathrm{Null}(\boldsymbol{R}_k)$ due to the projection $\boldsymbol{P}_k$.
Finally, when $\mathcal{R}(\boldsymbol{R}_1)=\mathcal{R}(\boldsymbol{R}_2)$, the projectors coincide.
Since $\boldsymbol{R}_1\succeq \boldsymbol{R}_2\succeq \boldsymbol{0}$ implies
$\big(\boldsymbol{I}+c\boldsymbol{R}_1\big)^{-1}\preceq \big(\boldsymbol{I}+c\boldsymbol{R}_2\big)^{-1}$ for $c>0$,
pre- or post-multiplying by the common projector yields $\boldsymbol{E}_1\preceq \boldsymbol{E}_2$ and thus $\mathrm{NMSE}_1\le \mathrm{NMSE}_2$ after taking traces.
\end{IEEEproof}

\vspace{-3pt}
\begin{remark}  \label{rem:NMSE_analysis}
Lemma~\ref{lem:NMSE} shows that RA rotation can improve normalized estimation quality by enlarging $\boldsymbol{R}_k(\boldsymbol{F})$, e.g., steering antenna boresights toward dominant scatterer directions. However, achievable rate depends on both $\boldsymbol{\mu}_k(\boldsymbol{F})$ and multiuser coupling terms that involve cross-user statistics and the combiner structure. Hence the rotation that minimizes NMSE is generally not the same as the one maximizing sum rate, whose closed-form expressions are analyzed in Section \ref{sec:rate_analysis}.
\end{remark}
\begin{remark} The LMMSE estimator \eqref{eq:LMMSE_estimator} depends on $\boldsymbol{F}$ only through the statistics $\boldsymbol{\mu}_k(\boldsymbol{F})$ and $\boldsymbol{R}_k(\boldsymbol{F})$. In the proposed two-timescale framework, $\boldsymbol{F}$ is updated on a slow timescale based on statistical CSI, and the same statistically matched LMMSE estimator is applied across many coherence blocks.
\end{remark}

\vspace{-6pt}
\section{Two-Timescale Optimization Problem}
\label{sec:problem_formulation}

In this section, we formulate the large-timescale RA rotation design
based on the channel estimates obtained in \eqref{eq:LMMSE_estimator}.
The key feature of the considered two-timescale architecture is that the
short-term receive combiner is updated in each coherence block using the
instantaneous channel estimate, whereas the RA orientation matrix
$\boldsymbol{F}$ is optimized only from slowly varying statistical CSI.

Recall from \eqref{eq:uplink_signal} and \eqref{eq:rx_combining} that,
for a given coherence block and a fixed $\boldsymbol{F}$, the BS applies
a linear combining matrix
\begin{align}
  \boldsymbol{V}(\hat{\boldsymbol{H}})
  \triangleq
  \big[
    \boldsymbol{v}_1(\hat{\boldsymbol{H}}),\dots,
    \boldsymbol{v}_K(\hat{\boldsymbol{H}})
  \big]
  \in\mathbb{C}^{N\times K},
\end{align}
being a deterministic function of the estimated channel matrix
\begin{align}
  \hat{\boldsymbol{H}}(\boldsymbol{F})
  \triangleq
  \big[
    \hat{\boldsymbol{h}}_1(\boldsymbol{F}),\dots,
    \hat{\boldsymbol{h}}_K(\boldsymbol{F})
  \big].
\end{align}
The $k$-th column $\boldsymbol{v}_k(\hat{\boldsymbol{H}})$ is the combining
vector used to detect user $k$. Typical examples include the MRC receiver, which uses conjugate beamforming to the channel, and the wZF receiver, which eliminates multiuser interference. Given $\boldsymbol{v}_k(\hat{\boldsymbol{H}})$, the detection statistic for user $k$ is
\begin{align}
  \varsigma_k= &\sqrt{p_k} \boldsymbol{v}_k^{H}(\hat{\boldsymbol{H}})
     \boldsymbol{h}_k(\boldsymbol{F}) x_k + \sum_{i\neq k} \sqrt{p_i} 
     \boldsymbol{v}_k^{H}(\hat{\boldsymbol{H}})   \boldsymbol{h}_i(\boldsymbol{F}) x_i 
     \nonumber
     \\
     &+ \boldsymbol{v}_k^{H}(\hat{\boldsymbol{H}}) \boldsymbol{z}.
  \label{eq:rk_def}
\end{align}
The corresponding instantaneous SINR of user $k$, conditioned on the
channels and their estimates, is
\begin{align}
  \gamma_k\big(\boldsymbol{F},\hat{\boldsymbol{H}}\big)\!
  &\!\triangleq
  \frac{
    p_k 
    \big|
      \boldsymbol{v}_k^{H}(\hat{\boldsymbol{H}})
      \boldsymbol{h}_k(\boldsymbol{F})
    \big|^2
  }{
    \sum_{i\neq k}
      p_i 
      \big|
        \boldsymbol{v}_k^{H}(\hat{\boldsymbol{H}})
        \boldsymbol{h}_i(\boldsymbol{F})
      \big|^2
    \!+\! \sigma^2\big\|\boldsymbol{v}_k(\hat{\boldsymbol{H}})\big\|^2
  }. \!\!\!
  \label{eq:sinr_inst}
\end{align}

To obtain a tractable large-timescale metric, we first define, for a given channel estimate $\hat{\boldsymbol{H}}(\boldsymbol{F})$ and combining matrix $\boldsymbol{V}(\hat{\boldsymbol{H}})$, the UatF SINR of user $k$ as in \eqref{eq:sinr_uatf}, where the BS and users
treat the conditional mean of the effective channel gain as
deterministic and regard all remaining randomness as additional uncorrelated
Gaussian noise \cite{ref27,ref28}. This yields a tractable lower bound on the achievable uplink rates that depends only on the second-order statistics of
$\hat{\boldsymbol{h}}_k(\boldsymbol{F})$ and
$\tilde{\boldsymbol{h}}_k(\boldsymbol{F})$.
\begin{figure*}
\begin{align}
&\qquad \qquad \bar{\gamma}_k\left( \boldsymbol{F},\boldsymbol{\hat{H}} \right) = \frac{p_k \big|\mathbb{E}\big[ \boldsymbol{v}_{k}^H( \boldsymbol{\hat{H}}) \boldsymbol{h}_k( \boldsymbol{F} )  | \boldsymbol{\hat{H}}( \boldsymbol{F})\big] \big|^2}{p_k \text{Var}\big( \boldsymbol{v}_{k}^H( \boldsymbol{\hat{H}}) \boldsymbol{h}_k\left( \boldsymbol{F} \right)  | \boldsymbol{\hat{H}}( \boldsymbol{F} )\big) +\sum_{i\ne k}{p_i} \mathbb{E}\big[ |\boldsymbol{v}_{k}^H( \boldsymbol{\hat{H}}) \boldsymbol{h}_i (\boldsymbol{F}) |^2 | \boldsymbol{\hat{H}}( \boldsymbol{F}) \big] +\sigma ^2 \big|\boldsymbol{v}_k( \boldsymbol{\hat{H}} ) \big|^2},
\label{eq:sinr_uatf}
\\
&\overline{\ \ \ \ \ \ \ \ \ \ \ \ \ \ \ \ \ \ \ \ \ \ \ \ \ \ \ \ \ \ \ \ \ \ \ \ \ \ \ \ \ \ \ \ \ \ \ \ \ \ \ \ \ \ \ \ \ \ \ \ \ \ \ \ \ \ \ \ \ \ \ \ \ \ \ \ \ \ \ \ \ \ \ \ \ \ \ \ \ \ \ \ \ \ \ \ \ \ \ \ \ \ \ \ \ \ \ \ \ \ \ \ \ \ \ \ \ \ \ \ \ \ \ \ \ \ \ \ \ \ \ \ \ \ \ \ \ \ \ \ \ \ \ \ } \nonumber
\end{align}
\vspace{-36pt}
\end{figure*}
The expectation in \eqref{eq:sinr_uatf} is taken with respect to the joint distribution of the small-scale fading, the channel estimates, and the
estimation errors, all of which depend on $\boldsymbol{F}$ via the Rician
statistics $\big(\boldsymbol{\mu}_k(\boldsymbol{F}),\boldsymbol{R}_k(\boldsymbol{F})\big)$
and the estimator \eqref{eq:LMMSE_estimator}.
Let $\eta\triangleq 1-\tau_\mathrm{p}/T_{\mathrm{c}}$ denote the pre-log factor
accounting for the pilot training overhead. The large-timescale conditional achievable rate metric is then defined as
\vspace{-4pt}
\begin{align}
\bar{R}_k(\boldsymbol{F})= \eta \log_2\big(1 + \bar{\gamma}_k(\boldsymbol{F},\hat{\boldsymbol{H}})\big),
  \label{eq:Rk_def}
\end{align}
and the corresponding sum rate surrogate is
\vspace{-4pt}
\begin{align}
  \bar{R}_{\text{sum}}(\boldsymbol{F})
  &\triangleq
  \sum_{k\in \mathcal{K}} \bar{R}_k(\boldsymbol{F}).
  \label{eq:Rsum_def}
\end{align}
Under the two-timescale design, the short-term combining rule $\boldsymbol{V}(\hat{\boldsymbol{H}})$ is determined by estimated instantaneous CSI, while the RA-orientation design problem, based on large-timescale CSI, maximizes $\boldsymbol{F}$ over the feasible set $\mathcal{F}$ as below
\vspace{-6pt}
\begin{subequations}
\label{prob:P1_general}
\begin{align}
  \text{(P1)}:\,\,
  \max_{\boldsymbol{F}}& \,\, \mathbb{E}\Big[ \max_{{\boldsymbol{V}}(\hat{\boldsymbol{H}})}
  \,\,
  \bar{R}_{\text{sum}}(\boldsymbol{F}) \Big]
  \label{prob:P1_obj}
  \\
   \text{s.t.}\,\, &\boldsymbol{F}\in\mathcal{F}.
  \label{prob:P1_constraint_F}
\end{align}
\end{subequations}
(P1) is a highly non-convex, stochastic optimization problem. Specifically, the expectation in \eqref{prob:P1_obj} has no simple closed form in general and cannot be directly evaluated inside an iterative optimization loop, the feasible set $\mathcal{F}$ involves unit-norm and angle constraints that make the search space highly non-convex, and the objective function related to directional radiation patterns is sophisticated and non-concave. 

To address these challenges, we first avoid the stochastic optimization by determining the combining vectors via MRC and wZF, which also draw useful insights in the RA system design. Accordingly, we derive closed-form rate approximations for MRC and wZF combining in Section~\ref{sec:rate_analysis}, denoted by $\bar{R}_k^{\mathrm{MRC}}(\boldsymbol{F})$ and $\bar{R}_k^{\mathrm{wZF}}(\boldsymbol{F})$, respectively, that depend only on statistics $\big(\boldsymbol{\mu}_k(\boldsymbol{F}),\boldsymbol{R}_k(\boldsymbol{F})\big)$ and induced LMMSE covariances. Substituting these approximations into \eqref{prob:P1_general} yields the surrogate problems with the only optimization variable $\boldsymbol{F}$
\vspace{-4pt}
\begin{subequations}
\label{prob:P2_MRC_wZF}
\begin{align}
  \text{(P2-MRC)}:\quad
  \max_{\boldsymbol{F}\in\mathcal{F}}
  &\quad
  \sum_{k\in \mathcal{K}} \bar{R}_k^{\mathrm{MRC}}(\boldsymbol{F}),
  \label{prob:P2_MRC}
  \\
  \text{(P2-wZF)}:\quad
  \max_{\boldsymbol{F}\in\mathcal{F}}
  &\quad
  \sum_{k\in \mathcal{K}} \bar{R}_k^{\mathrm{wZF}}(\boldsymbol{F}),
  \label{prob:P2_wZF}
\end{align}
\end{subequations}
which are still non-convex but piecewise-smooth constrained maximization problems over a product of spherical caps, serving as receiver-structure specializations of (P1).

\vspace{-6pt}
\begin{proposition}
\label{thm:single_user}
Assume $K=1$ and the LoS-dominant channel $\boldsymbol{R}_1(\boldsymbol{F})= \boldsymbol{0}$. Then $\hat{\boldsymbol{h}}_1(\boldsymbol{F})=\boldsymbol{h}_1(\boldsymbol{F})=\boldsymbol{\mu}_1(\boldsymbol{F})$ and MRC/wZF coincide. The large-timescale objective is maximized by independently maximizing $G(\boldsymbol{f}_n,\boldsymbol{s}_{1,n})$ for each $n$, hence
\vspace{-6pt}
\begin{align}
\boldsymbol{f}_n^\star=\Pi_{\mathcal{F}_n}(\boldsymbol{s}_{1,n}),\qquad \forall n\in\mathcal{N},
\label{eq:single_user_solution}
\end{align}
where $\Pi_{\mathcal{F}_n}(\cdot)$ denotes projection onto the spherical cap \eqref{eq:feasible_region}.
\end{proposition}
\begin{IEEEproof}
With $\boldsymbol{R}_1= \boldsymbol{0}$, the dependence on $\boldsymbol{F}$ is only through $\|\boldsymbol{\mu}_1(\boldsymbol{F})\|^2=\sum_n | \mu_{1,n}(\boldsymbol{F})|^2$, and $|\mu_{1,n}(\boldsymbol{F})|^2$ increases monotonically with $G(\boldsymbol{f}_n,\boldsymbol{s}_{1,n})$. Since constraints are separable across $n$, each $\boldsymbol{f}_n$ is optimized by maximizing $\boldsymbol{f}_n^T\boldsymbol{s}_{1,n}$ over $\mathcal{F}_n$, which is exactly \eqref{eq:single_user_solution}.
\end{IEEEproof}
\vspace{-3pt}
\begin{remark} 
Proposition \ref{thm:single_user} isolates the only regime in which the RA orientation design reduces to independent per-element gain maximization via boresight alignment, as rotation affects the objective only through the directional gain of each RA in the single-user pure-LoS case. This simple structure no longer holds in the general multiuser setting, where rotations jointly reshape the multipath channel statistics, multiuser interference, and CSI error. This motivates the receiver-specific large-timescale rate analysis for MRC and wZF in Section \ref{sec:rate_analysis}.
\end{remark}

\vspace{-9pt}
\section{Achievable Rate Analysis for MRC and wZF} \label{sec:rate_analysis}

In this section, we derive a closed-form UatF-based rate lower bound for MRC and a closed-form statistical surrogate for wZF. According to \eqref{eq:hk_distribution},   \eqref{eq:hat_h_distribution}, and \eqref{eq:tilde_h_distribution}, we introduce the second-order moment matrices for compactness
\begin{align}
  \!\!\boldsymbol{\Sigma}_{h,k}(\boldsymbol{F})
  &\!\triangleq\!
\mathbb{E}\big[\boldsymbol{h}_k(\boldsymbol{F})\boldsymbol{h}_k^H(\boldsymbol{F})\big] \!=\!  \boldsymbol{R}_k(\boldsymbol{F})
  \!+  \!\boldsymbol{\mu}_k(\boldsymbol{F})\boldsymbol{\mu}_k^H(\boldsymbol{F}),\!\!
  \label{eq:Sigma_h_def}
  \\
  \!\!\boldsymbol{\Sigma}_{\hat{h},k}(\boldsymbol{F})
  &\!\triangleq\!
  \mathbb{E}\big[
    \hat{\boldsymbol{h}}_k(\boldsymbol{F})
    \hat{\boldsymbol{h}}_k^H(\boldsymbol{F})
  \big]\!=\!  \boldsymbol{C}_{\!\hat{h}_k}\!(\boldsymbol{F})
  \!+\!  \boldsymbol{\mu}_k(\boldsymbol{F})\boldsymbol{\mu}_k^H(\boldsymbol{F}).\!\!\!
  \label{eq:Sigma_hhat_def}
\end{align}
Clearly, we have
$\boldsymbol{\Sigma}_{h,k}(\boldsymbol{F})
 = \boldsymbol{\Sigma}_{\hat{h},k}(\boldsymbol{F}) + \boldsymbol{C}_{\mathrm{e},k}(\boldsymbol{F})$.
Throughout this section, all expectations are taken over the joint
distribution of the actual channel, its LMMSE estimate, and the
estimation error, which all implicitly depend on
$\boldsymbol{F}$.

\vspace{-3pt}
\begin{remark} 
Unlike conventional arrays with fixed element patterns, $\boldsymbol{F}$ affects the RA main-lobe direction and reshapes the channel statistics through \eqref{eq:element_gain}. This dependence will propagate through the LMMSE estimator and the UatF bounds, creating nontrivial coupling between rotation, estimation error, and achievable rates under MRC and wZF reception.
\end{remark}

\vspace{-12pt}
\subsection{Closed-Form UatF Expression for MRC Receiver}

The MRC vector for user $k$ based on LMMSE estimates is
\begin{align}
  \boldsymbol{v}_k^{\mathrm{MRC}}(\hat{\boldsymbol{H}})
  \triangleq
  \hat{\boldsymbol{h}}_k(\boldsymbol{F}).
  \label{eq:MRC_combiner_def}
\end{align}
Substituting \eqref{eq:MRC_combiner_def} into \eqref{eq:sinr_uatf}, an unconditional UatF bound for the SINR of user $k$ under MRC is given by \eqref{eq:gamma_MRC_general}. 
\begin{figure*}
\begin{align}
&\qquad \quad \quad \quad \quad \bar{\gamma}_k^{\mathrm{MRC}}(\boldsymbol{F})
  \triangleq
  \frac{
    p_k 
    \big|
      \mathbb{E}\big[
        (\boldsymbol{v}_k^{\mathrm{MRC}})^H
        \boldsymbol{h}_k(\boldsymbol{F})
      \big]
    \big|^2
  }{
    \sum_{i=1}^K
      p_i 
      \mathbb{E}\big[
        \big|
          (\boldsymbol{v}_k^{\mathrm{MRC}})^H
          \boldsymbol{h}_i(\boldsymbol{F})
        \big|^2
      \big]
    - p_k
      \big|
        \mathbb{E}\big[
          (\boldsymbol{v}_k^{\mathrm{MRC}})^H
          \boldsymbol{h}_k(\boldsymbol{F})
        \big]
      \big|^2
    + \sigma^2
      \mathbb{E}\big[
        \|\boldsymbol{v}_k^{\mathrm{MRC}}\|^2
      \big]
  }.
  \label{eq:gamma_MRC_general}
\\
&\overline{\ \ \ \ \ \ \ \ \ \ \ \ \ \ \ \ \ \ \ \ \ \ \ \ \ \ \ \ \ \ \ \ \ \ \ \ \ \ \ \ \ \ \ \ \ \ \ \ \ \ \ \ \ \ \ \ \ \ \ \ \ \ \ \ \ \ \ \ \ \ \ \ \ \ \ \ \ \ \ \ \ \ \ \ \ \ \ \ \ \ \ \ \ \ \ \ \ \ \ \ \ \ \ \ \ \ \ \ \ \ \ \ \ \ \ \ \ \ \ \ \ \ \ \ \ \ \ \ \ \ \ \ \ \ \ \ \ \ \ \ \ \ \ \ } \nonumber
\end{align}
\vspace{-36pt}
\end{figure*}
For the useful signal term in the numerator of \eqref{eq:gamma_MRC_general}, since $\boldsymbol{h}_k(\boldsymbol{F})=\hat{\boldsymbol{h}}_k(\boldsymbol{F})+\tilde{\boldsymbol{h}}_k(\boldsymbol{F})$ with $\hat{\boldsymbol{h}}_k$ and $\tilde{\boldsymbol{h}}_k$ independent, we
have
\begin{subequations}
\begin{align}
  \mathbb{E}\big[
    (\boldsymbol{v}_k^{\mathrm{MRC}})^H
    \boldsymbol{h}_k(\boldsymbol{F})
  \big]
  & =
  \mathbb{E}\big[
    \|\hat{\boldsymbol{h}}_k(\boldsymbol{F})\|^2
  \big]
  \\
  &\!\!\!\!\!\!\!\!\!\!\!\!\!\!\!\!\!\!\!\!\!\!\!\!\!=
  \mathrm{tr}\big(\boldsymbol{C}_{\hat{h}_k}(\boldsymbol{F})\big)
  +
  \big\|
    \boldsymbol{\mu}_k(\boldsymbol{F})
  \big\|^2\triangleq\alpha_k(\boldsymbol{F}),
  \label{eq:MRC_signal_mean}
\end{align}
\end{subequations}
and then the numerator of \eqref{eq:gamma_MRC_general} becomes
\begin{align}
  p_k
  \big|
    \mathbb{E}\big[
      (\boldsymbol{v}_k^{\mathrm{MRC}})^H
      \boldsymbol{h}_k(\boldsymbol{F})
    \big]
  \big|^2
  =
  p_k \alpha_k^2(\boldsymbol{F}).
  \label{eq:MRC_S_term}
\end{align}
For the noise term, clearly we have
\begin{align}
  \mathbb{E}\big[
    \|\boldsymbol{v}_k^{\mathrm{MRC}}\|^2
  \big]
  =
  \mathbb{E}\big[
    \|\hat{\boldsymbol{h}}_k(\boldsymbol{F})\|^2
  \big]
  =
  \alpha_k(\boldsymbol{F}).
  \label{eq:MRC_combiner_norm} 
\end{align}
For interference terms when $i\neq k$, independence of channel estimates across users yields
\begin{align}
  \mathbb{E}\big[ \big|
      (\boldsymbol{v}_k^{\mathrm{MRC}})^H
      \boldsymbol{h}_i(\boldsymbol{F})
    \big|^2 \big]
  &=
  \mathbb{E}\big[
    \hat{\boldsymbol{h}}_k^H(\boldsymbol{F})
    \boldsymbol{h}_i(\boldsymbol{F})
    \boldsymbol{h}_i^H(\boldsymbol{F})
    \hat{\boldsymbol{h}}_k(\boldsymbol{F})
  \big]   \nonumber
  \\
  &=
  \mathrm{tr}\big(
    \boldsymbol{\Sigma}_{h,i}(\boldsymbol{F}) 
    \boldsymbol{\Sigma}_{\hat{h},k}(\boldsymbol{F})
  \big) \nonumber
    \\
  &\triangleq \varTheta_{i,k}(\boldsymbol{F}).
  \label{eq:MRC_interf_i_neq_k}
\end{align}
For $i=k$, due to the composition \eqref{eq:error_def}, 
$\mathbb{E}[\tilde{\boldsymbol{h}}_k(\boldsymbol{F})]=\boldsymbol{0}$, and the
independence between $\hat{\boldsymbol{h}}_k(\boldsymbol{F})$ and $\tilde{\boldsymbol{h}}_k(\boldsymbol{F})$, we obtain
\begin{subequations}
\begin{align}
&\mathbb{E}\big[
    \big|
      (\boldsymbol{v}_k^{\mathrm{MRC}})^H
      \boldsymbol{h}_k(\boldsymbol{F})
    \big|^2
  \big]  \nonumber
 \\
 & =  \mathbb{E}\big[
    \|\hat{\boldsymbol{h}}_k(\boldsymbol{F})\|^4
  \big] +
  \mathbb{E}\big[
    \big|
      \hat{\boldsymbol{h}}_k^H(\boldsymbol{F})
      \tilde{\boldsymbol{h}}_k(\boldsymbol{F})
    \big|^2
  \big], 
  \\
 & =  \alpha_k^2(\boldsymbol{F})  +  \mathrm{tr}\big( \boldsymbol{C}_{\hat{h}_k}^2(\boldsymbol{F})
  \big) + 2\boldsymbol{\mu}_k^H(\boldsymbol{F})
    \boldsymbol{C}_{\hat{h}_k}(\boldsymbol{F})
    \boldsymbol{\mu}_k(\boldsymbol{F}) \nonumber
  \\
 & \quad   
    +\mathrm{tr}\big( \boldsymbol{C}_{\mathrm{e},k}(\boldsymbol{F}) 
    \boldsymbol{\Sigma}_{\hat{h},k}(\boldsymbol{F}) \big)
  \label{eq:MRC_selfinterf}
\end{align}
\end{subequations}
where we use fourth-moment identities for complex Gaussian vectors. 
Collecting the terms \eqref{eq:MRC_combiner_norm},  \eqref{eq:MRC_interf_i_neq_k}, and \eqref{eq:MRC_selfinterf}, the total interference-plus-noise term in the denominator of \eqref{eq:gamma_MRC_general} is
\begin{align}
  &I_k^{\mathrm{MRC}}(\boldsymbol{F}) =p_k\varPhi_k(\boldsymbol{F})+  \sum_{i\neq k}   p_i \varTheta_{i,k}(\boldsymbol{F})
  + \sigma^2 \alpha_k(\boldsymbol{F}),
  \label{eq:MRC_I_term}
\end{align}
where we denote the self-interference terms by
\begin{align}
\varPhi_k(\boldsymbol{F})\triangleq &\mathbb{E}\big[
    \big|
      (\boldsymbol{v}_k^{\mathrm{MRC}})^H
      \boldsymbol{h}_k(\boldsymbol{F})
    \big|^2
  \big] -  \big|
    \mathbb{E}\big[
      (\boldsymbol{v}_k^{\mathrm{MRC}})^H
      \boldsymbol{h}_k(\boldsymbol{F})
    \big]
  \big|^2 \nonumber
  \\
  = &  
  \mathrm{tr}\big( \boldsymbol{C}_{\hat{h}_k}^2(\boldsymbol{F})
  \big) + 2\boldsymbol{\mu}_k^H(\boldsymbol{F})
    \boldsymbol{C}_{\hat{h}_k}(\boldsymbol{F})
    \boldsymbol{\mu}_k(\boldsymbol{F})  
    \nonumber
  \\
   &   +\mathrm{tr}\big( \boldsymbol{C}_{\mathrm{e},k}(\boldsymbol{F}) 
    \boldsymbol{\Sigma}_{\hat{h},k}(\boldsymbol{F}) \big),
    \label{eq:MRC_selfinterf_neat}
\end{align}
Substituting \eqref{eq:MRC_S_term} and \eqref{eq:MRC_I_term} into
\eqref{eq:gamma_MRC_general} yields the closed-form UatF SINR for MRC
\begin{align}
  \bar{\gamma}_k^{\mathrm{MRC}}(\boldsymbol{F})
  &=
  \frac{p_k\alpha_k^2(\boldsymbol{F})}{
    I_k^{\mathrm{MRC}}(\boldsymbol{F}) }.
  \label{eq:gamma_MRC_closedform}
\end{align}
The corresponding large-timescale rate surrogate is
\begin{align}
  \bar{R}_k^{\mathrm{MRC}}(\boldsymbol{F})
  \triangleq
  \eta 
  \log_2\!\big(
    1 + \bar{\gamma}_k^{\mathrm{MRC}}(\boldsymbol{F})
  \big),
  \label{eq:rate_MRC_closedform}
\end{align}
and the sum rate is $\bar{R}_{\text{sum}}^{\mathrm{MRC}}(\boldsymbol{F})
  \triangleq   \sum_{k=1}^K \bar{R}_k^{\mathrm{MRC}}(\boldsymbol{F})$.

\begin{remark}
\label{prop:mrc_structure}
It is observed from \eqref{eq:MRC_I_term} and \eqref{eq:gamma_MRC_closedform} that the rotation matrix $\boldsymbol{F}$ affects the MRC rate through four coupled mechanisms:
(i) improving the coherent signal aggregation or the desired estimated-channel energy via $\alpha_k(\boldsymbol{F})$,
(ii) reducing the estimation-error-induced self-interference or signal-leakage via $\varPhi_k(\boldsymbol{F})$,
(iii) reducing multi-user interference between the combining directions via $\varTheta_{i,k}(\boldsymbol{F})$, and
(iv) reducing post-combining noise amplification via $\sigma^2\alpha_k(\boldsymbol{F})$. 
Hence, MRC-oriented rotation is driven neither by NMSE minimization alone as depicted in Remark \ref{rem:NMSE_analysis} nor by simply desired signal power maximization alone in Proposition \ref{thm:single_user}.
\end{remark}

\vspace{-12pt}
\subsection{Closed-Form Statistical Surrogate for wZF Receiver}
\label{subsec:zf_uatf}

For MRC, the expectation in the UatF expression can be evaluated in closed form. For wZF, however, the conditional UatF SINR depends on the non-isotropic estimation errors caused by orientation-dependent gains, whose outer expectation is difficult to evaluate exactly. We therefore construct a closed-form large-timescale surrogate in the following. Assume 
$K\leq N$ and $\hat{\boldsymbol{H}}(\boldsymbol{F})$ has full column rank. Consider ZF-type combiners satisfying
\begin{align}
(\boldsymbol{V}^{\mathrm{ZF}})^H\hat{\boldsymbol{H}}(\boldsymbol{F})=\boldsymbol{I}_K.
\label{eq:wZF_property}
\end{align}
Using $\boldsymbol{h}_i(\boldsymbol{F})\!=\!\hat{\boldsymbol{h}}_i(\boldsymbol{F})+\tilde{\boldsymbol{h}}_i(\boldsymbol{F})$ and the independence between $\tilde{\boldsymbol{h}}_i(\boldsymbol{F})$ and $\hat{\boldsymbol{H}}(\boldsymbol{F})$, under the wZF property \eqref{eq:wZF_property}, we have
\begin{subequations}
\begin{align}
& \big(\boldsymbol{v}_k^{\mathrm{wZF}}\big)^H\boldsymbol{h}_k(\boldsymbol{F})
 =  1 + \big(\boldsymbol{v}_k^{\mathrm{wZF}}\big)^H\tilde{\boldsymbol{h}}_k(\boldsymbol{F}),  
\\ 
&\big(\boldsymbol{v}_k^{\mathrm{wZF}}\big)^H\boldsymbol{h}_i(\boldsymbol{F})
  = \big(\boldsymbol{v}_k^{\mathrm{wZF}}\big)^H\tilde{\boldsymbol{h}}_i(\boldsymbol{F}),
\\
&  \mathbb{E} \left[
    \big(\boldsymbol{v}_k^{\mathrm{wZF}}\big)^H\boldsymbol{h}_k(\boldsymbol{F})
    ~\big|~
    \hat{\boldsymbol{H}}(\boldsymbol{F})
  \right]=  1,
\\
&  \mathrm{Var}\big(
  (\boldsymbol{v}_k^{\mathrm{wZF}})^H\boldsymbol{h}_k(\boldsymbol{F})
    \big|  \hat{\boldsymbol{H}}(\boldsymbol{F})
  \big)
  =
  \big(\boldsymbol{v}_k^{\mathrm{wZF}}\big)^H\boldsymbol{C}_{\mathrm{e},k}(\boldsymbol{F})\boldsymbol{v}_k^{\mathrm{wZF}},
\\
&\mathbb{E}\big[
    \big| (\boldsymbol{v}_k^{\mathrm{wZF}})^H\boldsymbol{h}_i(\boldsymbol{F})
    \big|^2
    \big|  \hat{\boldsymbol{H}}(\boldsymbol{F})
  \big]=
  \big(\boldsymbol{v}_k^{\mathrm{wZF}}\big)^H\boldsymbol{C}_{\mathrm{e},i}(\boldsymbol{F})\boldsymbol{v}_k^{\mathrm{wZF}}.
\end{align}
\end{subequations}
Hence, the UatF SINR \eqref{eq:sinr_uatf} is transformed to
\begin{align}
\!\!\bar{\gamma}_k^{\mathrm{wZF}}\!\big(\boldsymbol{F},\hat{\boldsymbol{H}}\big)
\!=\!
\frac{p_k}{
\sum_{i=1}^K p_i (\boldsymbol{v}_k^{\mathrm{wZF}})^H\boldsymbol{C}_{e,i}(\boldsymbol{F})\boldsymbol{v}_k^{\mathrm{wZF}}
\!+\!\sigma^2\|\boldsymbol{v}_k^{\mathrm{wZF}}\|^2}.\!\!\!
\label{eq:gamma_wZF_uatf_cond}
\end{align}
Define the colored effective noise as
\begin{align}
\boldsymbol{Z}(\boldsymbol{F})
\triangleq
\sigma^2\boldsymbol{I}_N+\sum_{i=1}^K p_i\,\boldsymbol{C}_{e,i}(\boldsymbol{F}).
\label{eq:Z_def}
\end{align}
\begin{proposition}
\label{prop:wzf_opt}
Assume $K\leq N$ and $\hat{\boldsymbol{H}}(\boldsymbol{F})$ has full column rank, and thus $\hat{\boldsymbol{H}}^H(\boldsymbol{F})\boldsymbol{Z}^{-1}(\boldsymbol{F})\hat{\boldsymbol{H}}(\boldsymbol{F})$ is invertible; otherwise, a regularized version should be adopted. Then, among all ZF-type combiners satisfying \eqref{eq:wZF_property}, under colored effective noise, the one maximizing \eqref{eq:gamma_wZF_uatf_cond} for each user is the wZF \cite{ref29}
\begin{align}
\boldsymbol{V}^{\mathrm{wZF}}(\hat{\boldsymbol{H}})
\!=\!
\boldsymbol{Z}^{-1}(\boldsymbol{F})\hat{\boldsymbol{H}}(\boldsymbol{F})
\big(\hat{\boldsymbol{H}}^H(\boldsymbol{F})\boldsymbol{Z}^{-1}(\boldsymbol{F})\hat{\boldsymbol{H}}(\boldsymbol{F})\big)^{-1}, \!\!\!
\label{eq:wZF_def}
\end{align}
and the resulting conditional SINR is
\begin{align}
\bar{\gamma}_k^{\mathrm{wZF}}\big(\boldsymbol{F},\hat{\boldsymbol{H}}\big)
=
\frac{p_k}{
\Big[\big(\hat{\boldsymbol{H}}^H(\boldsymbol{F})\boldsymbol{Z}^{-1}(\boldsymbol{F})\hat{\boldsymbol{H}}(\boldsymbol{F})\big)^{-1}\Big]_{k,k}
}.
\label{eq:wZF_uatf_cond}
\end{align}
\end{proposition}
\begin{IEEEproof}
For each $k$, maximizing \eqref{eq:gamma_wZF_uatf_cond} is equivalent to minimizing $\boldsymbol{v}_k^H\boldsymbol{Z}\boldsymbol{v}_k$ subject to $\boldsymbol{v}_k^H\hat{\boldsymbol{H}}=\boldsymbol{e}_k^T$. This strictly convex quadratic program yields \eqref{eq:wZF_def} by KKT conditions, and substituting into \eqref{eq:gamma_wZF_uatf_cond} gives \eqref{eq:wZF_uatf_cond}.
\end{IEEEproof}
To enable large-timescale optimization using statistical CSI only, we adopt a first-order statistical surrogate as below, by replacing the random whitened Gram matrix with \cite{ref30}
\begin{align}
\hat{\boldsymbol{H}}^H(\boldsymbol{F})\boldsymbol{Z}^{-1}(\boldsymbol{F})\hat{\boldsymbol{H}}(\boldsymbol{F})
&\approx
\mathbb{E}\Big[
\hat{\boldsymbol{H}}^H(\boldsymbol{F})\boldsymbol{Z}^{-1}(\boldsymbol{F})\hat{\boldsymbol{H}}(\boldsymbol{F})
\Big] \nonumber
\\
&\triangleq\bar{\boldsymbol{S}}(\boldsymbol{F}).
\label{eq:Sbar_def}
\end{align}
Let $\boldsymbol{M}(\boldsymbol{F})=[\boldsymbol{\mu}_1(\boldsymbol{F}),\dots,\boldsymbol{\mu}_K(\boldsymbol{F})]$ collect the LoS components. Since $\hat{\boldsymbol{h}}_k\sim\mathcal{CN}(\boldsymbol{\mu}_k,\boldsymbol{C}_{\hat{h}_k})$, we obtain
\begin{align}
\bar{\boldsymbol{S}}(\boldsymbol{F})
\!=&
\mathrm{diag}\big(
\mathrm{tr}\big(\boldsymbol{Z}^{-1}\!(\boldsymbol{F})\boldsymbol{C}_{\hat{h}_1}\!(\boldsymbol{F})\big),
\dots,
\mathrm{tr}\big(\boldsymbol{Z}^{-1}\!(\boldsymbol{F})\boldsymbol{C}_{\hat{h}_K}\!(\boldsymbol{F})\big)
\big) \nonumber
\\
&+\boldsymbol{M}^H(\boldsymbol{F})\boldsymbol{Z}^{-1}(\boldsymbol{F})\boldsymbol{M}(\boldsymbol{F}).
\label{eq:Sbar_closed}
\end{align}
Then the surrogate SINR and rate are
\begin{align}
&\bar{\gamma}_k^{\mathrm{wZF}}(\boldsymbol{F})
\triangleq
\frac{p_k}{\big[\bar{\boldsymbol{S}}^{-1}(\boldsymbol{F})\big]_{k,k}},
\label{eq:wZF_SINR_stat}
\\
&\bar{R}_k^{\mathrm{wZF}}(\boldsymbol{F})
\triangleq
\eta\log_2\!\big(1+\bar{\gamma}_k^{\mathrm{wZF}}(\boldsymbol{F})\big).
\label{eq:wZF_rate_stat}
\end{align}
and the sum rate is $\bar{R}_{\text{sum}}^{\mathrm{wZF}}(\boldsymbol{F})
  \triangleq   \sum_{k=1}^K \bar{R}_k^{\mathrm{wZF}}(\boldsymbol{F})$.
\begin{remark}
\label{rem:isotropic}
If we consider the isotropic estimation error across antennas for all $i$ as  
\begin{align}
\boldsymbol{C}_{\mathrm{e},i}(\boldsymbol{F}) \approx \sigma_{\mathrm{e},i}^2(\boldsymbol{F})\,\boldsymbol{I}_N,
\quad \forall i\in\mathcal{K},
\label{eq:Ce_isotropic_assumption}
\end{align}
\eqref{eq:gamma_wZF_uatf_cond} reduces to a classical wZF noise-enhancement form
\begin{align}
\bar{\gamma}_k^{\mathrm{wZF}}\big(\boldsymbol{F},\hat{\boldsymbol{H}}\big)
=
\frac{
p_k
}{
\big(
\sigma^2 + \sum_{i=1}^K p_i\sigma_{\mathrm{e},i}^2(\boldsymbol{F})
\big)
\big[\boldsymbol{G}^{-1}(\boldsymbol{F})\big]_{k,k}
},
\label{eq:gamma_wZF_uatf_isotropic}
\end{align}
where $\boldsymbol{G}(\boldsymbol{F})\triangleq \hat{\boldsymbol{H}}(\boldsymbol{F})^H\hat{\boldsymbol{H}}(\boldsymbol{F})$. In this case, one may further invoke massive-MIMO moment-matched approximation
\begin{align}
\!&\mathbb{E}\big[\big[\boldsymbol{G}^{-1}(\boldsymbol{F})\big]_{k,k}\big]
\approx
\frac{1}{(N-K)N\hat{\beta}_k(\boldsymbol{F})}
\Big[\hat{\boldsymbol{\Sigma}}^{-1}(\boldsymbol{F})\Big]_{k,k},
\label{eq:DE_Ginv_kk}
\\
\!&\hat{\boldsymbol{\Sigma}}(\boldsymbol{F})
\!\triangleq\!
\hat{\boldsymbol{\Lambda}}_1(\boldsymbol{F})
\!+\!
\frac{1}{N\!-\!K}
\hat{\boldsymbol{\Lambda}}_2^H(\boldsymbol{F})
\bar{\boldsymbol{H}}^H(\boldsymbol{F})
\bar{\boldsymbol{H}}(\boldsymbol{F})
\hat{\boldsymbol{\Lambda}}_2(\boldsymbol{F}),\!\!
\label{eq:Sigma_hat_def}
\end{align}
with 
leading to the approximation
\begin{align}
  \bar{\gamma}_k^{\mathrm{wZF}}(\boldsymbol{F})
  \triangleq
  \frac{
    p_k (N-K)N\hat{\beta}_k(\boldsymbol{F})
  }{
    \big(
      \sigma^2 + \sum_{i=1}^K p_i\sigma_{\mathrm{e},i}^2(\boldsymbol{F})
    \big)
    \big[\hat{\boldsymbol{\Sigma}}^{-1}(\boldsymbol{F})\big]_{k,k}
  },
  \label{eq:gamma_wZF_closedform_LMMSE}
\end{align}
with $\hat{\beta}_k(\boldsymbol{F})
\triangleq \frac{1}{N}\mathbb{E}\big[\|\hat{\boldsymbol{h}}_k(\boldsymbol{F})\|^2\big]$, $\bar{\boldsymbol{h}}_k(\boldsymbol{F})
\triangleq
\sqrt{N}
\frac{\boldsymbol{\mu}_k(\boldsymbol{F})}{\|\boldsymbol{\mu}_k(\boldsymbol{F})\|}$, $\bar{\boldsymbol{H}}(\boldsymbol{F})
\triangleq
\big[
\bar{\boldsymbol{h}}_1(\boldsymbol{F}), \dots, \bar{\boldsymbol{h}}_K(\boldsymbol{F})
\big]$, $\hat{\boldsymbol{\Lambda}}_1(\boldsymbol{F})
\triangleq
\big(\hat{\boldsymbol{\Omega}}(\boldsymbol{F})+\boldsymbol{I}_K\big)^{-1}$, $\hat{\boldsymbol{\Lambda}}_2(\boldsymbol{F})
\triangleq
\big(\hat{\boldsymbol{\Omega}}(\boldsymbol{F})\hat{\boldsymbol{\Lambda}}_1(\boldsymbol{F})\big)^{1/2}$, $\hat{\boldsymbol{\Omega}}(\boldsymbol{F})=\mathrm{diag}(\hat{\kappa}_1,\ldots,\hat{\kappa}_K)$, and $\hat{\kappa}_k(\boldsymbol{F})
\triangleq
\frac{\|\boldsymbol{\mu}_k(\boldsymbol{F})\|^2}{\mathrm{tr}(\boldsymbol{C}_{\hat{h}_k}(\boldsymbol{F}))}$. In RA systems, however, \eqref{eq:Ce_isotropic_assumption} is generally violated because $\boldsymbol{F}$ induces antenna-dependent gains; hence we use \eqref{eq:Sbar_closed}--\eqref{eq:wZF_rate_stat} as the main ZF-type surrogate.
\end{remark}

\begin{corollary} \label{thm:wZF_schur}
Let $\bar{\boldsymbol{S}}(\boldsymbol{F})$
in \eqref{eq:Sbar_closed} be partitioned with respect to user $k$ as
\begin{align}
\bar{\boldsymbol{S}}(\boldsymbol{F})
=
\begin{bmatrix}
\bar{s}_{k,k}(\boldsymbol{F}) &
\bar{\boldsymbol{s}}_{k,-k}^H(\boldsymbol{F})
\\
\bar{\boldsymbol{s}}_{k,-k}(\boldsymbol{F}) &
\bar{\boldsymbol{S}}_{-k,-k}(\boldsymbol{F})
\end{bmatrix},
\label{eq:Sbar_block}
\end{align}
where
\begin{align}
\bar{s}_{k,k}(\boldsymbol{F})
&=
\boldsymbol{\mu}_k^H(\boldsymbol{F})
\boldsymbol{Z}^{-1}(\boldsymbol{F})
\boldsymbol{\mu}_k(\boldsymbol{F})
+
\tr\big(
\boldsymbol{Z}^{-1}(\boldsymbol{F})
\boldsymbol{C}_{\hat{h}_k}(\boldsymbol{F})
\big)
\nonumber\\
&=
\big\|
\boldsymbol{Z}^{-1/2}(\boldsymbol{F})
\boldsymbol{\mu}_k(\boldsymbol{F})
\big\|^2
\nonumber \\
&\quad +
\tr\big(
\boldsymbol{Z}^{-1/2}(\boldsymbol{F})
\boldsymbol{C}_{\hat{h}_k}(\boldsymbol{F})
\boldsymbol{Z}^{-1/2}(\boldsymbol{F})
\big),
\label{eq:wZF_self_gain}
\\
\bar{\boldsymbol{s}}_{k,-k}(\boldsymbol{F})
&=
\Big[
\boldsymbol{\mu}_i^H(\boldsymbol{F})
\boldsymbol{Z}^{-1}(\boldsymbol{F})
\boldsymbol{\mu}_k(\boldsymbol{F})
\Big]_{i\neq k}.
\label{eq:wZF_cross_corr}
\end{align}
Then the exact representation of the SINR surrogate with wZF combiner in \eqref{eq:wZF_SINR_stat} can be written as
\begin{align}
\bar{\gamma}_k^{\mathrm{wZF}}(\boldsymbol{F})=
p_k
\big(
\bar{s}_{k,k}(\boldsymbol{F})
-
\bar{\boldsymbol{s}}_{k,-k}^H(\boldsymbol{F})
\bar{\boldsymbol{S}}_{-k,-k}^{-1}(\boldsymbol{F})
\bar{\boldsymbol{s}}_{k,-k}(\boldsymbol{F})
\big).
\label{eq:wZF_schur_gamma}
\end{align}
Therefore, RA rotation $\boldsymbol{F}$ affects wZF rate through
two explicit mechanisms:
(i) it increases the error-aware useful signal strength, 
$\bar{s}_{k,k}(\boldsymbol{F})$, 
while (ii) reducing the error-aware inter-user coupling, or the multiuser interference, $\bar{\boldsymbol{s}}_{k,-k}^H
\bar{\boldsymbol{S}}_{-k,-k}^{-1}
\bar{\boldsymbol{s}}_{k,-k}$. Since $\boldsymbol{Z}(\boldsymbol{F})$ itself depends on $\{\boldsymbol{C}_{e,i}(\boldsymbol{F})\}$,
estimation errors affect both the useful signal term and the interference term simultaneously.
\end{corollary}

\begin{IEEEproof}
\eqref{eq:wZF_schur_gamma} follows from the Schur-complement identity
\begin{align}
[\bar{\boldsymbol{S}}^{-1}]_{k,k}^{-1}
=
\bar{s}_{k,k}-\bar{\boldsymbol{s}}_{k,-k}^H\bar{\boldsymbol{S}}_{-k,-k}^{-1}\bar{\boldsymbol{s}}_{k,-k}.
\end{align}
Since
$\bar{\gamma}_k^{\mathrm{wZF}}(\boldsymbol{F})
=
\frac{p_k}{[\bar{\boldsymbol{S}}^{-1}(\boldsymbol{F})]_{k,k}}$ in \eqref{eq:wZF_SINR_stat}, we arrive at \eqref{eq:wZF_schur_gamma}.
\end{IEEEproof}

\section{Rotation Optimization}
\label{sec:rotation_optimization}

In this section, we develop efficient algorithms for solving the RA rotation design problems (P2-MRC) and (P2-wZF) in \eqref{prob:P2_MRC_wZF} based on large-timescale statistical CSI. Note that Proposition~\ref{prop:wzf_opt} has established wZF as the conditionally optimal ZF-type combiner under the colored effective-noise covariance \eqref{eq:Z_def}. In the following, we use the wZF rate surrogate \eqref{eq:wZF_rate_stat} and MRC rate surrogate \eqref{eq:rate_MRC_closedform} as the objective function in the algorithmic design, respectively
\begin{subequations}
\label{prob:P3_MRC_wZF}
\begin{align}
\text{(P3-MRC)}:\quad
\max_{\boldsymbol{F}\in\mathcal{F}}
&\quad
\bar{R}_{\text{sum}}^{\mathrm{MRC}}(\boldsymbol{F})\triangleq
\sum_{k=1}^K
\bar{R}_k^{\mathrm{MRC}}(\boldsymbol{F}),
\label{prob:P3_MRC_obj_new}
\\
\text{(P3-wZF)}:\quad
\max_{\boldsymbol{F}\in\mathcal{F}}
&\quad
\bar{R}_{\text{sum}}^{\mathrm{wZF}}(\boldsymbol{F})\triangleq
\sum_{k=1}^K
\bar{R}_k^{\mathrm{wZF}}(\boldsymbol{F}).
\label{prob:P3_wZF_obj_new}
\end{align}
\end{subequations}
The difficulty lies in the highly non-convex dependence of the closed-form rate surrogates
$\bar{R}_k^{\mathrm{MRC}}(\boldsymbol{F})$ and $\bar{R}_k^{\mathrm{wZF}}(\boldsymbol{F})$ on the RA rotations $\boldsymbol{F}$ via the Rician means $\boldsymbol{\mu}_k(\boldsymbol{F})$, covariances $\boldsymbol{R}_k(\boldsymbol{F})$, and LMMSE error covariances $\boldsymbol{C}_{\mathrm{e},k}(\boldsymbol{F})$. 
However, the resulting optimization problems are piecewise-smooth constrained maximization problems over the product feasible set $\mathcal{F}=\mathcal{F}_1\times
\cdots\times \mathcal{F}_N$, where each $\mathcal{F}_n$ is a spherical cap.
This structure naturally motivates a projected-gradient method:
the gradient step is taken in the ambient Euclidean space, then projected
onto the tangent space of the unit sphere to preserve unit norm to first
order, and finally projected onto the spherical cap to satisfy the tilt
constraint. Therefore, we adopt a projected gradient ascent framework combined with the closed-form rate approximations in Section~\ref{sec:rate_analysis} for MRC and wZF receivers \cite{ref31}. 

\vspace{-12pt}
\subsection{Derivatives of the Channel Statistics}
\label{subsec:orientation_derivatives}

The dependence of the objectives on $\boldsymbol{F}$ is fully inherited from
the statistics $\{\boldsymbol{\mu}_k(\boldsymbol{F}),\boldsymbol{R}_k(\boldsymbol{F}),
\boldsymbol{C}_{\hat{h}_k}(\boldsymbol{F}),\boldsymbol{C}_{\mathrm{e},k}(\boldsymbol{F})\}$.
Since the $n$-th RA orientation vector $\boldsymbol{f}_n$ only affects the
$n$-th entry of each $\boldsymbol{\mu}_k(\boldsymbol{F})$ and the $n$-th row of
each $\boldsymbol{B}_k(\boldsymbol{F})$, the corresponding derivatives admit
a sparse structure, which is crucial for obtaining implementable gradients.
From \eqref{eq:element_gain}, given a direction $\boldsymbol{s}$, the directional gain derivative is
\begin{align}
\nabla_{\boldsymbol{f}_n}G(\boldsymbol{f}_n,\boldsymbol{s})
=
\begin{cases}
2bG_0(\boldsymbol{f}_n^T\boldsymbol{s})^{2b-1}\boldsymbol{s}, &
\boldsymbol{f}_n^T\boldsymbol{s}>0,\\
\boldsymbol{0}, & \boldsymbol{f}_n^T\boldsymbol{s}\le 0.
\end{cases}
\label{eq:gain_grad}
\end{align}
At the boundary $\boldsymbol{f}_n^T\boldsymbol{s}=0$, which occurs on a
measure-zero set, we use the zero subgradient in implementation.
Let $\boldsymbol{\delta}_n\in\mathbb{R}^{N\times 1}$ denote the $n$-th
canonical basis vector, i.e., the vector whose $n$-th entry is one and all the others are zero. From \eqref{eq:hk_n_LoS}, the $n$-th entry of
$\boldsymbol{\mu}_k(\boldsymbol{F})$ is
\begin{align}
\mu_{k,n}(\boldsymbol{F})
=
\sqrt{\varrho\frac{G_0}{4\pi r_{k,n}^2}}
[\boldsymbol{f}_n^T\boldsymbol{s}_{k,n}]_+^{\,b}
e^{-j\frac{2\pi}{\lambda}r_{k,n}}.
\label{eq:mu_kn_closed}
\end{align}
For $m\in\{\mathrm{x},\mathrm{y},\mathrm{z}\}$, using \eqref{eq:gain_grad}, we obtain
\begin{align}
&\frac{\partial \mu_{k,n}(\boldsymbol{F})} {\partial [\boldsymbol{f}_n]_m}
=
\begin{cases}
\dfrac{b\,\mu_{k,n}(\boldsymbol{F})}{\boldsymbol{f}_n^T\boldsymbol{s}_{k,n}}
[\boldsymbol{s}_{k,n}]_m,
& \boldsymbol{f}_n^T\boldsymbol{s}_{k,n}>0,\\
\boldsymbol{0}, & \boldsymbol{f}_n^T\boldsymbol{s}_{k,n}\le 0.
\end{cases}
\label{eq:mu_kn_m}
\\
&\frac{\partial \boldsymbol{\mu}_k(\boldsymbol{F})}
{\partial [\boldsymbol{f}_n]_m}
=
\boldsymbol{\delta}_n\,\frac{\partial \mu_{k,n}(\boldsymbol{F})} {\partial [\boldsymbol{f}_n]_m}.
\label{eq:mu_vec_derivative}
\end{align}
Next, from \eqref{eq:hk_expression} and \eqref{eq:hk_n_NLoS}, the $(n,q)$-th
entry of $\boldsymbol{B}_k(\boldsymbol{F})$ is
\begin{align}
\![\boldsymbol{B}_k(\boldsymbol{F})]_{n,q}
\!=\!
\frac{\sqrt{\varrho\sigma_q G_0}}{\sqrt{4\pi}r_{q,n}d_{q,k}}
[\boldsymbol{f}_n^T\boldsymbol{s}_{q,n}]_+^{b}
e^{-j\frac{2\pi}{\lambda}(r_{q,n}\!+d_{q,k})}.\!\!\!
\label{eq:B_knq_closed}
\end{align}
Then, similar to \eqref{eq:mu_kn_m}, the derivative of \eqref{eq:B_knq_closed} is 
\begin{align}
&\frac{\partial [\boldsymbol{B}_k(\boldsymbol{F})]_{n,q}} {\partial [\boldsymbol{f}_n]_m}
\!=\!
\begin{cases}
\dfrac{b\,[\boldsymbol{B}_k(\boldsymbol{F})]_{n,q}}{\boldsymbol{f}_n^T\boldsymbol{s}_{q,n}}
[\boldsymbol{s}_{q,n}]_m,
& \!\!\!\boldsymbol{f}_n^T\boldsymbol{s}_{q,n}>0,\\
0, & \!\!\!\boldsymbol{f}_n^T\boldsymbol{s}_{q,n}\le 0.
\end{cases} \!\!
\label{eq:B_knq_m}
\\
&\frac{\partial \boldsymbol{B}_k(\boldsymbol{F})}
{\partial [\boldsymbol{f}_n]_m}
=
\boldsymbol{\delta}_n \big(\boldsymbol{b}_{k,n}^{(m)}(\boldsymbol{F})\big)^T,
\label{eq:B_mat_derivative}
\end{align}
where $\boldsymbol{b}_{k,n}^{(m)}(\boldsymbol{F})$ collects \eqref{eq:B_knq_m} over $Q$ scattering paths
\begin{align}
\boldsymbol{b}_{k,n}^{(m)}(\boldsymbol{F})
\triangleq
\big[
\frac{\partial [\boldsymbol{B}_k(\boldsymbol{F})]_{n,1}} {\partial [\boldsymbol{f}_n]_m},\ldots,
\frac{\partial [\boldsymbol{B}_k(\boldsymbol{F})]_{n,Q}} {\partial [\boldsymbol{f}_n]_m}
\big]^T.
\label{eq:b_kn_m_vec}
\end{align}
Since $\boldsymbol{R}_k(\boldsymbol{F})=\boldsymbol{B}_k(\boldsymbol{F})
\boldsymbol{B}_k^H(\boldsymbol{F})$, it follows that
\begin{align}
\frac{\partial \boldsymbol{R}_k(\boldsymbol{F})}
{\partial [\boldsymbol{f}_n]_m}
=
\boldsymbol{\delta}_n \big(\boldsymbol{b}_{k,n}^{(m)}\big)^T
\boldsymbol{B}_k^H
+
\boldsymbol{B}_k
\big(\boldsymbol{b}_{k,n}^{(m)}\big)^*
\boldsymbol{\delta}_n^T.
\label{eq:R_kn_m}
\end{align}
Recall $a_k\triangleq \tau_{\mathrm{p}}p_k^{\mathrm{tr}}$ and denote 
$\boldsymbol{A}_k(\boldsymbol{F})
\triangleq
a_k\boldsymbol{R}_k(\boldsymbol{F})+\sigma^2\boldsymbol{I}_N$.
Using \eqref{eq:est_covariance} and 
$\mathrm{d}\boldsymbol{A}_k^{-1}
=
-\boldsymbol{A}_k^{-1}(\mathrm{d}\boldsymbol{A}_k)\boldsymbol{A}_k^{-1}$,
the derivative of the LMMSE covariance is
\begin{align}
&\frac{\partial \boldsymbol{C}_{\hat{h}_k}(\boldsymbol{F})}
{\partial [\boldsymbol{f}_n]_m} \nonumber
\\
&=
\!a_k\!\frac{\partial \boldsymbol{R}_k(\boldsymbol{F})} {\partial [\boldsymbol{f}_n]_m}\boldsymbol{A}_k^{-1}(\boldsymbol{F})\boldsymbol{R}_k(\boldsymbol{F})
\!+\!
a_k\boldsymbol{R}_k(\boldsymbol{F})\boldsymbol{A}_k^{-1}(\boldsymbol{F})\frac{\partial \boldsymbol{R}_k(\boldsymbol{F})} {\partial [\boldsymbol{f}_n]_m}
\nonumber\\
&\quad
-a_k^2\boldsymbol{R}_k(\boldsymbol{F})\boldsymbol{A}_k^{-1}(\boldsymbol{F})
\frac{\partial \boldsymbol{R}_k(\boldsymbol{F})} {\partial [\boldsymbol{f}_n]_m}
\boldsymbol{A}_k^{-1}(\boldsymbol{F})\boldsymbol{R}_k(\boldsymbol{F}).
\label{eq:Chat_kn_m}
\end{align}
Accordingly, we have 
\begin{align}
\frac{\partial \boldsymbol{C}_{\mathrm{e},k}(\boldsymbol{F})}
{\partial [\boldsymbol{f}_n]_m}
=
\frac{\partial \boldsymbol{R}_k(\boldsymbol{F})} {\partial [\boldsymbol{f}_n]_m}
-
\frac{\partial \boldsymbol{C}_{\hat{h}_k}(\boldsymbol{F})} {\partial [\boldsymbol{f}_n]_m}.
\end{align}
Finally, define the derivatives of the second-order moment matrices
\begin{align}
\frac{\partial \boldsymbol{\Sigma}_{h,k}(\boldsymbol{F})}
{\partial [\boldsymbol{f}_n]_m}
=&
\frac{\partial \boldsymbol{R}_k(\boldsymbol{F})} {\partial [\boldsymbol{f}_n]_m}
+\boldsymbol{\delta}_n\frac{\partial \mu_{k,n}(\boldsymbol{F})}
{\partial [\boldsymbol{f}_n]_m}\boldsymbol{\mu}_k^H(\boldsymbol{F}) \nonumber
\\
&+\boldsymbol{\mu}_k(\boldsymbol{F})\Big(\frac{\partial \mu_{k,n}(\boldsymbol{F})}
{\partial [\boldsymbol{f}_n]_m}\Big)^{*}\boldsymbol{\delta}_n^T,
\label{eq:Sigma_h_kn_m}
\\
\frac{\partial \boldsymbol{\Sigma}_{\hat{h},k}(\boldsymbol{F})}
{\partial [\boldsymbol{f}_n]_m}
=&
\frac{\partial \boldsymbol{C}_{\hat{h}_k}(\boldsymbol{F})}
{\partial [\boldsymbol{f}_n]_m} 
+\boldsymbol{\delta}_n\frac{\partial \mu_{k,n}(\boldsymbol{F})}
{\partial [\boldsymbol{f}_n]_m}\boldsymbol{\mu}_k^H(\boldsymbol{F})\nonumber
\\
&+\boldsymbol{\mu}_k(\boldsymbol{F})\Big(\frac{\partial \mu_{k,n}(\boldsymbol{F})}
{\partial [\boldsymbol{f}_n]_m}\Big)^{*}\boldsymbol{\delta}_n^T.
\label{eq:Sigma_hhat_kn_m}
\end{align}

\subsection{Gradient of the MRC Sum-Rate Surrogate}
\label{subsec:MRC_gradient}

For the MRC surrogate SINR \eqref{eq:gamma_MRC_closedform} with \eqref{eq:MRC_I_term}, using \eqref{eq:MRC_signal_mean} and \eqref{eq:mu_vec_derivative} combined with \eqref{eq:Chat_kn_m}, we obtain
\begin{align}
\!\frac{\partial \alpha_k(\boldsymbol{F})} {\partial [\boldsymbol{f}_n]_m}
\!=\!
2\Rea\Big\{
\mu_{k,n}^*(\boldsymbol{F})\frac{\partial \mu_{k,n}(\boldsymbol{F})} {\partial [\boldsymbol{f}_n]_m}
\!\Big\}
\!+\!
\tr\Big(
\frac{\partial \boldsymbol{C}_{\hat{h}_k}(\boldsymbol{F})} {\partial [\boldsymbol{f}_n]_m}
\!\Big). \!\!
\label{eq:alpha_kn_m}
\end{align}
Next, for the multiuser interference term $\varTheta_{i,k}(\boldsymbol{F})$ in \eqref{eq:MRC_interf_i_neq_k}, 
\begin{align}
\frac{\partial \varTheta_{i,k}(\boldsymbol{F})} {\partial [\boldsymbol{f}_n]_m}
=
\tr\Big(
\frac{\partial \boldsymbol{\Sigma}_{h,i}(\boldsymbol{F})}
{\partial [\boldsymbol{f}_n]_m}
\boldsymbol{\Sigma}_{\hat{h},k}(\boldsymbol{F}) +
\boldsymbol{\Sigma}_{h,i}(\boldsymbol{F})
\frac{\partial \boldsymbol{\Sigma}_{\hat{h},k}(\boldsymbol{F})} {\partial [\boldsymbol{f}_n]_m}
\Big).
\label{eq:Theta_kn_m}
\end{align}
For the self-interference term
$\varPhi_k(\boldsymbol{F})$ in \eqref{eq:MRC_selfinterf_neat}, 
\begin{align}
&\frac{\partial \varPhi_k(\boldsymbol{F})} {\partial [\boldsymbol{f}_n]_m} \nonumber
\\
&=
2\tr\!\Big(
\boldsymbol{C}_{\hat{h}_k}(\boldsymbol{F})
\frac{\partial \boldsymbol{C}_{\hat{h}_k}(\boldsymbol{F})} {\partial [\boldsymbol{f}_n]_m}
\Big)
+
2\boldsymbol{\mu}_k^H(\boldsymbol{F})
\frac{\partial \boldsymbol{C}_{\hat{h}_k}(\boldsymbol{F})} {\partial [\boldsymbol{f}_n]_m}
\boldsymbol{\mu}_k(\boldsymbol{F})
\nonumber\\
&\quad
+
4\Rea\Big\{
\Big(\frac{\partial \mu_{k,n}(\boldsymbol{F})}
{\partial [\boldsymbol{f}_n]_m}\Big)^{*}(\boldsymbol{F})
\big[
\boldsymbol{C}_{\hat{h}_k}(\boldsymbol{F})
\boldsymbol{\mu}_k(\boldsymbol{F})
\big]_n
\Big\}
\nonumber\\
&\quad+
\tr\Big(
\frac{\partial \boldsymbol{C}_{\mathrm{e},k}(\boldsymbol{F})} {\partial [\boldsymbol{f}_n]_m}
\boldsymbol{\Sigma}_{\hat{h},k}(\boldsymbol{F})
\!+\!
\boldsymbol{C}_{\mathrm{e},k}(\boldsymbol{F})
\frac{\partial \boldsymbol{\Sigma}_{\hat{h},k}(\boldsymbol{F})} {\partial [\boldsymbol{f}_n]_m}
\Big).\!\!\!
\label{eq:Phi_tilde_kn_m}
\end{align}
Hence, the derivative of the MRC denominator is
\begin{align}
\frac{\partial I_k^{\mathrm{MRC}}(\boldsymbol{F})}
{\partial [\boldsymbol{f}_n]_m}
=
p_k\frac{\partial \varPhi_k(\boldsymbol{F})} {\partial [\boldsymbol{f}_n]_m}
+
\sum_{i\neq k}p_i\frac{\partial \varTheta_{i,k}(\boldsymbol{F})} {\partial [\boldsymbol{f}_n]_m}
+
\sigma^2\frac{\partial \alpha_k(\boldsymbol{F})} {\partial [\boldsymbol{f}_n]_m}.
\label{eq:I_MRC_kn_m}
\end{align}
Therefore, the derivative of the MRC SINR surrogate is
\begin{align}
\frac{\partial \bar{\gamma}_k^{\mathrm{MRC}}(\boldsymbol{F})}
{\partial [\boldsymbol{f}_n]_m}
\!=\!
p_k\frac{
2\alpha_k(\boldsymbol{F})I_k^{\mathrm{MRC}}(\boldsymbol{F})
\frac{\partial \alpha_k(\boldsymbol{F})} {\partial [\boldsymbol{f}_n]_m}
\!-\!
\alpha_k^2(\boldsymbol{F})
\frac{\partial I_k^{\mathrm{MRC}}(\boldsymbol{F})} {\partial [\boldsymbol{f}_n]_m}
}{
\big(I_k^{\mathrm{MRC}}(\boldsymbol{F})\big)^2
}.
\label{eq:gamma_MRC_kn_m}
\end{align}
Stacking the three partial derivatives gives the gradient
\begin{align}
\nabla_{\boldsymbol{f}_n} \bar{R}_{\text{sum}}^{\mathrm{MRC}}(\boldsymbol{F})
=
\frac{\eta}{\ln 2}
\sum_{k=1}^K
\frac{1}{1+\bar{\gamma}_k^{\mathrm{MRC}}(\boldsymbol{F})}
\begin{bmatrix}
\frac{\partial \bar{\gamma}_k^{\mathrm{MRC}}}{\partial [\boldsymbol{f}_n]_\mathrm{x}}\\
\frac{\partial \bar{\gamma}_k^{\mathrm{MRC}}}{\partial [\boldsymbol{f}_n]_\mathrm{y}}\\
\frac{\partial \bar{\gamma}_k^{\mathrm{MRC}}}{\partial [\boldsymbol{f}_n]_\mathrm{z}}
\end{bmatrix}.
\label{eq:grad_MRC_final}
\end{align}

\vspace{-12pt}
\subsection{Gradient of the wZF Sum-Rate Surrogate}
\label{subsec:wZF_gradient}
For the wZF surrogate \eqref{eq:wZF_SINR_stat}, define $\Psi_k(\boldsymbol{F})
\triangleq
[\bar{\boldsymbol{S}}^{-1}(\boldsymbol{F})]_{k,k}$, then $\bar{\gamma}_k^{\mathrm{wZF}}(\boldsymbol{F})
=
\frac{p_k}{\Psi_k(\boldsymbol{F})}$.
According to \eqref{eq:Z_def}, we have
\begin{align}
&\frac{\partial \boldsymbol{Z}(\boldsymbol{F})}
{\partial [\boldsymbol{f}_n]_m}
=
\sum_{i=1}^K p_i\frac{\partial \boldsymbol{C}_{\mathrm{e},i}(\boldsymbol{F})}
{\partial [\boldsymbol{f}_n]_m},
\\
&\frac{\partial \boldsymbol{Z}^{-1}(\boldsymbol{F})}
{\partial [\boldsymbol{f}_n]_m}
=
-\boldsymbol{Z}^{-1}(\boldsymbol{F})
\frac{\partial \boldsymbol{Z}(\boldsymbol{F})} {\partial [\boldsymbol{f}_n]_m}
\boldsymbol{Z}^{-1}(\boldsymbol{F}).
\label{eq:Zinv_kn_m}
\end{align}
To handle \eqref{eq:Sbar_closed}, we compute the derivative of the LoS matrix
\begin{align}
\frac{\partial \boldsymbol{M}(\boldsymbol{F})}
{\partial [\boldsymbol{f}_n]_m}
=
\Big[
\boldsymbol{\delta}_n\frac{\partial \mu_{1,n}(\boldsymbol{F})} {\partial [\boldsymbol{f}_n]_m},
\ldots,
\boldsymbol{\delta}_n\frac{\partial \mu_{K,n}(\boldsymbol{F})} {\partial [\boldsymbol{f}_n]_m}
\Big].
\label{eq:M_kn_m}
\end{align}
Further, we derive 
\begin{align}
\chi_{k,n}^{(m)}(\boldsymbol{F})
&\triangleq
\frac{\partial \tr\big(
\boldsymbol{Z}^{-1}(\boldsymbol{F})
\boldsymbol{C}_{\hat{h}_k}(\boldsymbol{F})
\big)}{\partial [\boldsymbol{f}_n]_m}
\\
&= 
\tr\Big(
\frac{\partial \boldsymbol{Z}^{-1}(\boldsymbol{F})} {\partial [\boldsymbol{f}_n]_m}
\boldsymbol{C}_{\hat{h}_k}(\boldsymbol{F})
+
\boldsymbol{Z}^{-1}(\boldsymbol{F})
\frac{\partial \boldsymbol{C}_{\hat{h}_k}(\boldsymbol{F})} {\partial [\boldsymbol{f}_n]_m}
 \Big).
\label{eq:tau_kn_m}
\end{align}
Let $\boldsymbol{D}_{n}^{(m)}(\boldsymbol{F})
\triangleq
\diag\!\big(
\chi_{1,n}^{(m)}(\boldsymbol{F}),\ldots,
\chi_{K,n}^{(m)}(\boldsymbol{F})
\big)$.
Then, differentiating \eqref{eq:Sbar_closed}, we obtain
\begin{align}
\frac{\partial \bar{\boldsymbol{S}}(\boldsymbol{F})}
{\partial [\boldsymbol{f}_n]_m}
&=
\boldsymbol{D}_{n}^{(m)}(\boldsymbol{F})
+
\Big(\frac{\partial \boldsymbol{M}(\boldsymbol{F})} {\partial [\boldsymbol{f}_n]_m}(\boldsymbol{F})\Big)^H
\boldsymbol{Z}^{-1}(\boldsymbol{F})
\boldsymbol{M}(\boldsymbol{F})
\nonumber\\
&\quad
+
\boldsymbol{M}^{H}(\boldsymbol{F})
\boldsymbol{Z}^{-1}(\boldsymbol{F})
\frac{\partial \boldsymbol{M}(\boldsymbol{F})} {\partial [\boldsymbol{f}_n]_m}
\nonumber\\
&\quad
+
\boldsymbol{M}^{H}(\boldsymbol{F})
\frac{\partial \boldsymbol{Z}^{-1}(\boldsymbol{F})} {\partial [\boldsymbol{f}_n]_m}
\boldsymbol{M}(\boldsymbol{F}).
\label{eq:Sbar_kn_m}
\end{align}

Using the matrix inverse identity
\begin{align}
\frac{\partial \bar{\boldsymbol{S}}^{-1}(\boldsymbol{F})}
{\partial [\boldsymbol{f}_n]_m}
=
-\bar{\boldsymbol{S}}^{-1}(\boldsymbol{F})
\frac{\partial \bar{\boldsymbol{S}}(\boldsymbol{F})} {\partial [\boldsymbol{f}_n]_m}
\bar{\boldsymbol{S}}^{-1}(\boldsymbol{F}),
\label{eq:Sbar_inv_derivative}
\end{align}
the derivative of $\Psi_k(\boldsymbol{F})$ is
\begin{align}
\frac{\partial \Psi_k(\boldsymbol{F})}
{\partial [\boldsymbol{f}_n]_m}
=
-
\Big[
\bar{\boldsymbol{S}}^{-1}(\boldsymbol{F})
\frac{\partial \bar{\boldsymbol{S}}(\boldsymbol{F})} {\partial [\boldsymbol{f}_n]_m}
\bar{\boldsymbol{S}}^{-1}(\boldsymbol{F})
\Big]_{k,k}.
\label{eq:Psi_kn_m}
\end{align}
Therefore, the derivative of the wZF SINR surrogate is
\begin{align}
\frac{\partial \bar{\gamma}_k^{\mathrm{wZF}}(\boldsymbol{F})}
{\partial [\boldsymbol{f}_n]_m}
=
\frac{p_k}{\Psi_k^2(\boldsymbol{F})}
\Big[
\bar{\boldsymbol{S}}^{-1}(\boldsymbol{F})
\frac{\partial \bar{\boldsymbol{S}}(\boldsymbol{F})} {\partial [\boldsymbol{f}_n]_m}
\bar{\boldsymbol{S}}^{-1}(\boldsymbol{F})
\Big]_{k,k}.
\label{eq:gamma_wZF_kn_m}
\end{align}
Hence, the gradient of the wZF sum-rate surrogate is
\begin{align}
\nabla_{\boldsymbol{f}_n} \bar{R}^{\mathrm{wZF}}_{\text{sum}}(\boldsymbol{F})
=
\frac{\eta}{\ln 2}
\sum_{k=1}^K
\frac{1}{1+\bar{\gamma}_k^{\mathrm{wZF}}(\boldsymbol{F})}
\begin{bmatrix}
\frac{\partial \bar{\gamma}_k^{\mathrm{wZF}}}{\partial [\boldsymbol{f}_n]_\mathrm{x}}\\
\frac{\partial \bar{\gamma}_k^{\mathrm{wZF}}}{\partial [\boldsymbol{f}_n]_\mathrm{y}}\\
\frac{\partial \bar{\gamma}_k^{\mathrm{wZF}}}{\partial [\boldsymbol{f}_n]_\mathrm{z}}
\end{bmatrix}.
\label{eq:grad_wZF_final}
\end{align}

\vspace{-12pt}
\subsection{Projected Gradient Ascent with Spherical-Cap Projection}
\label{subsec:pga}

Given either objective $\bar{R}^{\mathrm{BF}}_{\text{sum}}(\boldsymbol{F})$ with
$\mathrm{BF}\in\{\mathrm{MRC},\mathrm{wZF}\}$, define the gradient
\begin{align}
\!\!\nabla\!_{\boldsymbol{f}_n}\!\bar{R}^{\mathrm{BF}}_{\text{sum}}(\boldsymbol{F})
\!=\!\!
\bigg[\!\frac{\partial \bar{R}_{\text{sum}}^{\text{BF}}( \boldsymbol{F} )}{\partial \left[ \boldsymbol{f}_n \right] _{\text{x}}},\frac{\partial \bar{R}_{\text{sum}}^{\text{BF}}( \boldsymbol{F} )}{\partial \left[ \boldsymbol{f}_n \right] _{\text{y}}},\frac{\partial \bar{R}_{\text{sum}}^{\text{BF}}( \boldsymbol{F} )}{\partial \left[ \boldsymbol{f}_n \right] _{\text{z}}}\! \bigg]^T. \!\!\!
\label{eq:euclid_grad_def_new}
\end{align}
To preserve the unit-norm constraint, we first remove the radial component
\begin{align}
\boldsymbol{g}_n^{(\ell)}
=
\Big(
\boldsymbol{I}_3-\boldsymbol{f}_n^{(\ell)}(\boldsymbol{f}_n^{(\ell)})^T
\Big)
\nabla_{\boldsymbol{f}_n}
\bar{R}^{\mathrm{BF}}_{\text{sum}}(\boldsymbol{F}^{(\ell)}).
\label{eq:tangent_proj_update}
\end{align}
This is simply the Euclidean projection of the gradient onto the tangent space
of the sphere. Given a step size $\alpha^{(\ell)}>0$, the tentative update is
\begin{align}
\tilde{\boldsymbol{f}}_n^{(\ell+1)}
=
\frac{
\boldsymbol{f}_n^{(\ell)}+\alpha^{(\ell)}\boldsymbol{g}_n^{(\ell)}
}{
\big\|
\boldsymbol{f}_n^{(\ell)}+\alpha^{(\ell)}\boldsymbol{g}_n^{(\ell)}
\big\|
},
\label{eq:retraction}
\end{align}
which is then projected onto the spherical cap $\mathcal{F}_n$ via \eqref{eq:cone_projection},
\begin{figure*}
\begin{align}
\!&\qquad \qquad \qquad \boldsymbol{f}_n^{(\ell+1)}
=
\Pi_{\mathcal{F}_n}\big(\tilde{\boldsymbol{f}}_n^{(\ell+1)}\big)
=
\begin{cases}
\tilde{\boldsymbol{f}}_n^{(\ell+1)}, &
\arccos\!\big((\tilde{\boldsymbol{f}}_n^{(\ell+1)})^T\boldsymbol{e}_z\big)
\le \theta_{\max},
\\
\cos\theta_{\max}\boldsymbol{e}_z
+
\sin\theta_{\max}
\dfrac{
\boldsymbol{a}_n^{(\ell+1)}
}{
\|\boldsymbol{a}_n^{(\ell+1)}\|
}, & \text{otherwise},
\end{cases}
\label{eq:cone_projection}
\\
\!&\overline{\ \ \ \ \ \ \ \ \ \ \ \ \ \ \ \ \ \ \ \ \ \ \ \ \ \ \ \ \ \ \ \ \ \ \ \ \ \ \ \ \ \ \ \ \ \ \ \ \ \ \ \ \ \ \ \ \ \ \ \ \ \ \ \ \ \ \ \ \ \ \ \ \ \ \ \ \ \ \ \ \ \ \ \ \ \ \ \ \ \ \ \ \ \ \ \ \ \ \ \ \ \ \ \ \ \ \ \ \ \ \ \ \ \ \ \ \ \ \ \ \ \ \ \ \ \ \ \ \ \ \ \ \ \ \ \ \ \ \ \ \ \ \ \ \ } \nonumber
\end{align}
\vspace{-33pt}
\end{figure*}
where
\begin{align}
\boldsymbol{a}_n^{(\ell+1)}
\triangleq
\tilde{\boldsymbol{f}}_n^{(\ell+1)}
-
\big((\tilde{\boldsymbol{f}}_n^{(\ell+1)})^T\boldsymbol{e}_z\big)\boldsymbol{e}_z.
\label{eq:horizontal_component}
\end{align}
The step size is selected by Armijo backtracking: choose the largest
$\alpha^{(\ell)}=\rho^q\alpha_0$ with $\rho\in(0,1)$ such that
\begin{align}
\bar{R}^{\mathrm{BF}}_{\text{sum}}(\boldsymbol{F}^{(\ell+1)})
\ge
\bar{R}^{\mathrm{BF}}_{\text{sum}}(\boldsymbol{F}^{(\ell)})
+
c\,\alpha^{(\ell)}
\sum_{n=1}^N
\|\boldsymbol{g}_n^{(\ell)}\|^2,
\label{eq:armijo_rule}
\end{align}
where $c\in(0,1)$ is the Armijo parameter. The iterations stop when the norm of effective gradient is below a predetermined threshold $\epsilon >0$ as $
\max_{n\in\mathcal{N}}\|\boldsymbol{g}_n^{(\ell)}\|
\le \epsilon$ or the maximum number of iterations is reached.

\begin{algorithm}[t]
\caption{Projected Gradient Ascent for RA Rotation}
\label{alg:PGA}
\small
\begin{algorithmic}[1]
\STATE \textbf{Input:} Initial $\boldsymbol{F}^{(0)}\in\mathcal{F}$, tolerance $\epsilon$, maximum iteration number $L_{\max}$, receiver type $\mathrm{BF}\in\{\mathrm{MRC},\mathrm{wZF}\}$.
\FOR{$\ell=0,1,\ldots,L_{\max}-1$}
\STATE Compute $\nabla_{\boldsymbol{f}_n}\bar{R}^{\mathrm{BF}}_{\text{sum}}(\boldsymbol{F}^{(\ell)})$ for all $n$ using \eqref{eq:grad_MRC_final} if $\mathrm{BF}=\mathrm{MRC}$, or \eqref{eq:grad_wZF_final} if $\mathrm{BF}=\mathrm{wZF}$.
\STATE Compute projected ascent directions
$\boldsymbol{g}_n^{(\ell)}$ by \eqref{eq:tangent_proj_update}.
\IF{$\max_n\|\boldsymbol{g}_n^{(\ell)}\|\le \epsilon$}
\STATE \textbf{break}
\ENDIF
\STATE Choose $\alpha^{(\ell)}$ by Armijo backtracking using \eqref{eq:armijo_rule}.
\STATE Update $\tilde{\boldsymbol{f}}_n^{(\ell+1)}$ by \eqref{eq:retraction}, and project onto $\mathcal{F}_n$ via \eqref{eq:cone_projection}.
\STATE Form $\boldsymbol{F}^{(\ell+1)}=[\boldsymbol{f}_1^{(\ell+1)},\ldots,\boldsymbol{f}_N^{(\ell+1)}]$.
\ENDFOR
\STATE \textbf{Output:} $\boldsymbol{F}^{\star}=\boldsymbol{F}^{(\ell+1)}$.
\end{algorithmic}
\end{algorithm}

\vspace{-6pt}
\subsection{Convergence and Complexity}
\vspace{-3pt}
The above algorithm is summarized in Algorithm~\ref{alg:PGA}. 
Let $\bar{R}^{\mathrm{BF}}_{\text{sum}}(\boldsymbol{F})$ be either
$\bar{R}^{\mathrm{MRC}}(\boldsymbol{F})$ or $\bar{R}^{\mathrm{wZF}}(\boldsymbol{F})$.
Under Algorithm \ref{alg:PGA} with Armijo backtracking \eqref{eq:armijo_rule}, the sequence
$\{\bar{R}^{\mathrm{BF}}_{\text{sum}}(\boldsymbol{F}^{(\ell)})\}$ is non-decreasing, i.e., $\bar{R}^{\mathrm{BF}}_{\text{sum}}(\boldsymbol{F}^{(\ell+1)})
\ge
\bar{R}^{\mathrm{BF}}_{\text{sum}}(\boldsymbol{F}^{(\ell)})$
for all $\ell$. Since the objective is continuous and $\mathcal{F}$ is compact, the objective sequence is upper bounded and therefore convergent.
In addition, the algorithm terminates at a point satisfying the projected-gradient stopping criterion, which serves as a Clarke stationary point candidate for \eqref{prob:P3_MRC_obj_new} or \eqref{prob:P3_wZF_obj_new}.

The algorithm is implementation-friendly. The sparsity structure in
\eqref{eq:mu_vec_derivative} and \eqref{eq:B_mat_derivative} implies that
changing $\boldsymbol{f}_n$ affects only the $n$-th entry/row of the channel
statistics, which substantially reduces the complexity of gradient evaluation.
Assuming direct matrix inversions for the matrices $\boldsymbol{Z}(\boldsymbol{F})$, $\bar{\boldsymbol{S}}(\boldsymbol{F})$ and $\boldsymbol{A}(\boldsymbol{F})$ in
\eqref{eq:Z_def}, \eqref{eq:Sbar_def} and \eqref{eq:Chat_kn_m}, the dominant per-iteration
cost of the MRC design is $\mathcal{O}(KN^3+K^2N^2)$, while that of the wZF
design is $\mathcal{O}((K+1)N^3+K^2N^2+K^3)$ due to the additional inversions
of $\boldsymbol{Z}(\boldsymbol{F})$ and $\bar{\boldsymbol{S}}(\boldsymbol{F})$. Note that the above complexity orders assume that the common matrices $\{\boldsymbol{A}_k^{-1}(\boldsymbol{F})\}_{k=1}^K$, $\boldsymbol{Z}^{-1}(\boldsymbol{F})$, and $\bar{\boldsymbol{S}}^{-1}(\boldsymbol{F})$ are computed once per iteration and then reused across all element-wise gradient evaluations.
In practice, both designs converge within a modest number of iterations, and
the wZF is more expensive but typically delivers larger gains.

\vspace{-6pt}
\section{Numerical Results}
\vspace{-3pt}
\label{sec:numerical_results}

Unless otherwise specified, we consider a BS equipped with an $N=N_{\mathrm{row}}N_{\mathrm{col}}=2\times 4$ RA-enabled UPA operating at $f_c=6$~GHz with half-wavelength inter-element spacing, the number of users and scatterer clusters are $K=4$ and $Q=3$, respectively. The antenna directional parameter is set to $b=4$ with the corresponding maximum element gain $G_0=2(2b+1)=18$. The noise power is $\sigma^2=-80$ dBm, the path loss at reference distance 1 m is $-30$ dB, which is equal to set $\frac{\varrho}{4\pi}=10^{-3}~ \mathrm{m}^2$, the coherence block length is $T_c=200$, and the default maximum rotation angle is $\theta_{\max}=60^\circ$. We set $\tau_p = K$. All users transmit at the same maximum power in both the pilot and data phases, i.e., $p_k^{\mathrm{tr}}=p_k\triangleq p=20$ dBm, $\forall k\in\mathcal{K}$. The user locations are uniformly generated within the region with horizontal radius in $[0,300]$~m and height in $[100,200]$~m, while the scatterer clusters are generated within horizontal radius in $[0,350]$~m and height in $[50,250]$~m with $\sigma_q=\frac{100}{3} \mathrm{m}^2$. All curves are obtained by averaging over 3000 independent geometry realizations; for the curves labeled ``erg.'', which means ergodic, an additional block-level averaging over $100$ independent fast fading realizations is carried out for each geometry realization, whereas the curves labeled ``sur.'' are computed from the large-timescale closed-form surrogates derived in Section~\ref{sec:rate_analysis}. The label ``opt.~$\boldsymbol{F}$'' denotes the RA orientation matrix obtained by Algorithm~\ref{alg:PGA}; ``ran.~$\boldsymbol{F}$'' denotes random feasible rotations, where each boresight vector $f_n$ is independently sampled uniformly over the spherical cap $\mathcal{F}_n$; and ``fix.~$\boldsymbol{F}$'' in the insight figures refers to $\boldsymbol{f}_n=\boldsymbol{e}_z$, $\forall n\in\mathcal{N}$.

\setlength{\abovecaptionskip}{-3pt}
\begin{figure}[t]
\centering
\includegraphics[width=2.6in]{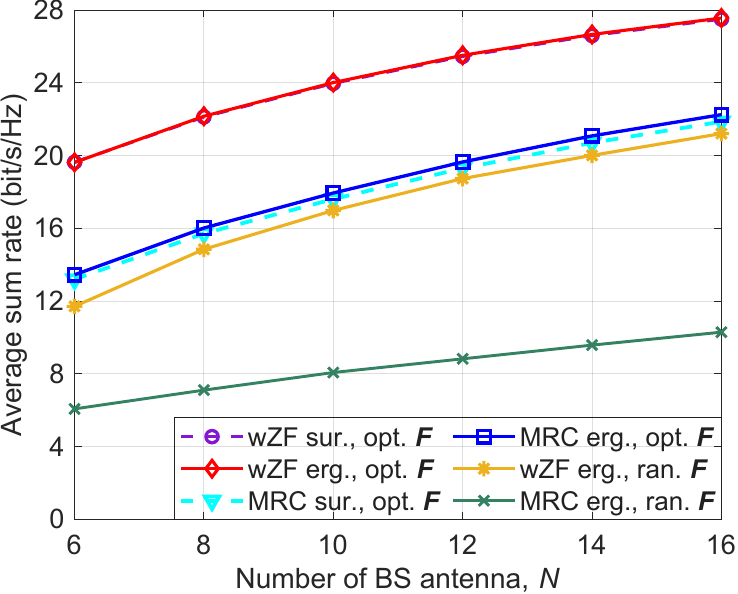}
\caption{Average sum rate versus the number of BS antennas.} 
\label{fig:N}
\vspace{-6pt}
\end{figure}

\begin{figure}[t]
\centering
\includegraphics[width=2.6in]{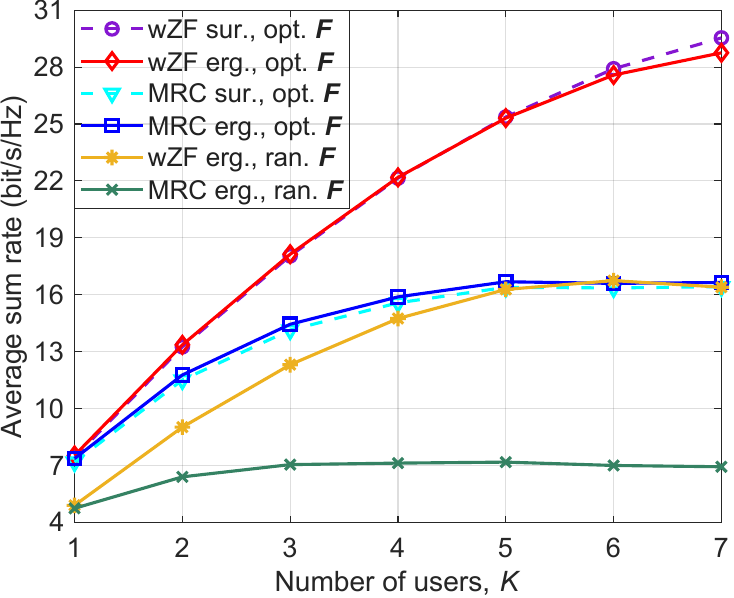}
\caption{Average sum rate versus the number of users.} 
\label{fig:K}
\vspace{-18pt}
\end{figure}

\begin{figure}[t]
\centering
\includegraphics[width=2.6in]{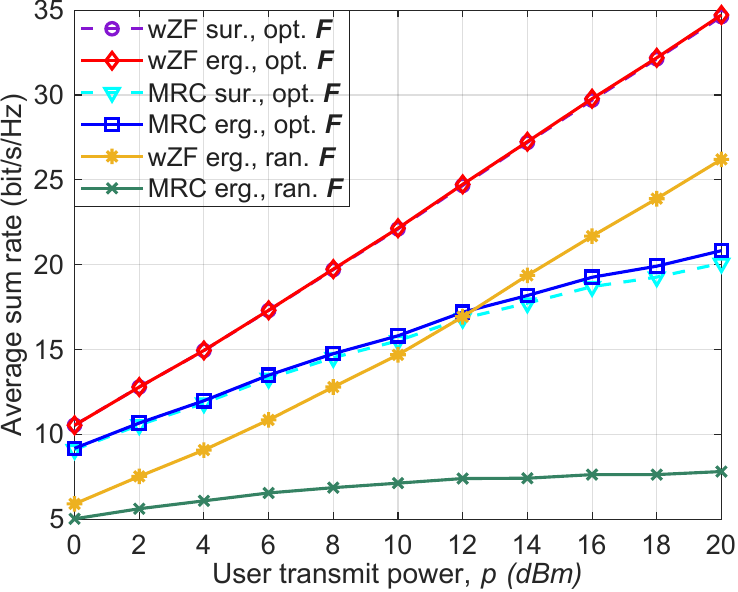}
\caption{Average sum rate versus the user transmit power.} 
\label{fig:P}
\vspace{-18pt}
\end{figure}

Figs.~\ref{fig:N}, \ref{fig:K}, and \ref{fig:P} show the average sum rate versus the number of BS antennas $N$, the number of users $K$, and the user transmit power $p$, respectively. A first and important observation is that, under the optimized rotation matrix, the surrogate curves and the corresponding ergodic curves are nearly indistinguishable for both MRC and wZF across all three figures. This confirms that the closed-form UatF surrogates in \eqref{eq:rate_MRC_closedform} and \eqref{eq:wZF_rate_stat} are sufficiently accurate for large-timescale RA design and, more importantly, that optimizing the surrogate objectives in \eqref{prob:P3_MRC_obj_new}-\eqref{prob:P3_wZF_obj_new} transfers to the actual block-level performance.
The superiority of the optimized RA design over random feasible rotations is also substantial. In Fig.~\ref{fig:N}, both optimized MRC and optimized wZF increase monotonically with $N$, but the wZF curve remains consistently above the MRC curve, indicating that RA rotation and interference-aware combining are complementary rather than competing mechanisms. In particular, increasing $N$ simultaneously strengthens the desired effective channels and enlarges the spatial DoF available to the colored-noise-aware wZF combiner. 

Fig.~\ref{fig:K} further highlights the different scaling laws of MRC and wZF. As the number of users increases, the optimized wZF sum rate continues to grow strongly, whereas the optimized MRC curve rises more slowly and saturates once the system enters an interference-limited regime. This behavior is consistent with the insights in Section~\ref{sec:rate_analysis}. The wZF surrogate exploits the inverse-whitened Gram matrix and thus benefits from RA-enabled user separability, whereas MRC remains governed by the coupled signal, interference, and noise terms in Remark~4. Therefore, once the user population becomes dense, the performance bottleneck of MRC is no longer signal strength but residual multiuser coupling, whereas wZF can still convert the additional users into multiplexing gain by jointly optimizing the orientations and the interference-aware combiner. Fig.~\ref{fig:P} conveys the same message from a different angle. The optimized wZF rate grows almost linearly with the transmit power in dB over the considered range, whereas the optimized MRC rate increases more slowly. This indicates that, after RA optimization, wZF behaves closer to a noise-limited regime over a wider operating region, while MRC transitions earlier into an interference-limited regime. The gap between optimized and random orientations also widens with $p$, showing that RA optimization becomes more valuable as the system moves from a noise-limited regime to a regime where control of interference geometry is critical.

\begin{figure}[t]
\centering
\includegraphics[width=2.6in]{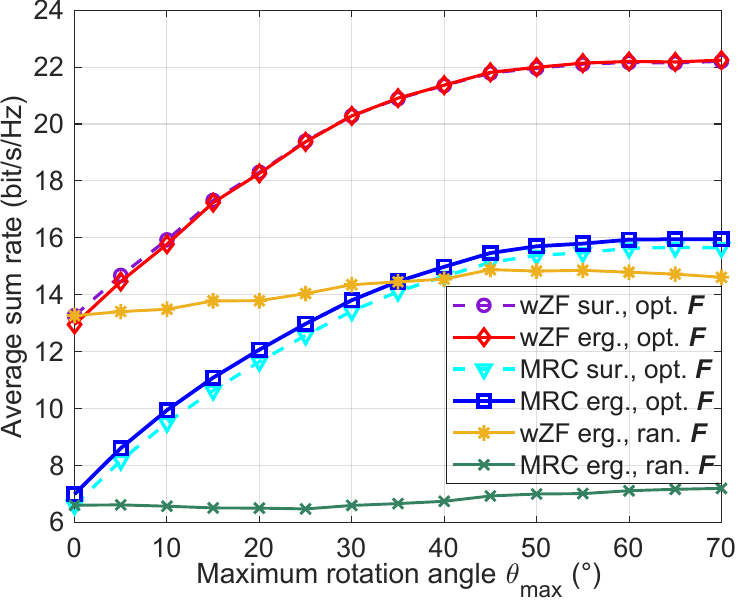}
\caption{Average sum rate versus the maximum rotation angle.} 
\label{fig:R}
\vspace{-6pt}
\end{figure}

\begin{figure}[t]
\centering
\includegraphics[width=2.6in]{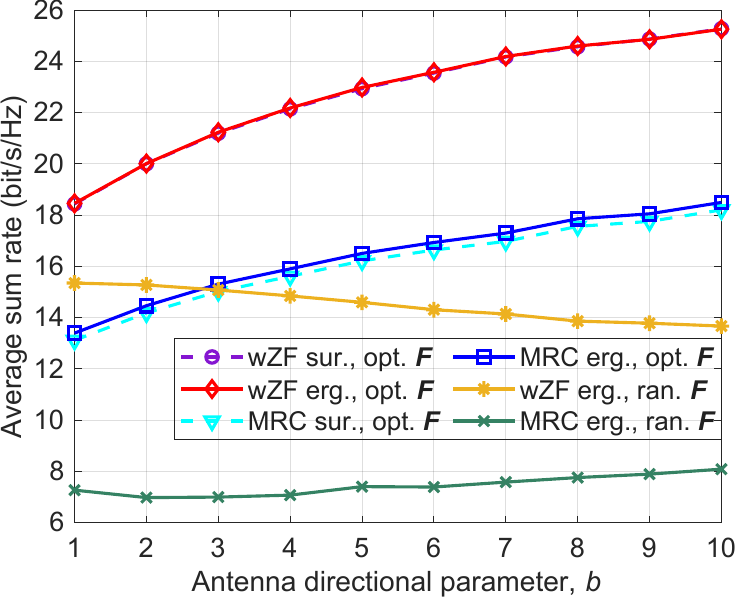}
\caption{Average sum rate versus the antenna directional parameter.} 
\label{fig:b}
\vspace{-18pt}
\end{figure}

We next examine the physical RA parameters. Fig.~\ref{fig:R} plots the average sum rate versus the maximum allowable rotation angle $\theta_{\max}$. For both receivers, enlarging the feasible spherical cap yields a pronounced rate improvement at small and moderate $\theta_{\max}$, followed by a clear saturation when $\theta_{\max}$ becomes sufficiently large. Beyond a certain angular freedom, further enlarging the mechanical rotation range produces only marginal gains, which suggests that the projection constraint is rarely active. Hence, a moderate rotation range, depending on the user and scatterer distributions, is generally sufficient to capture most of the RA benefit. Fig.~\ref{fig:b} shows the impact of the directional parameter $b$. Under the optimized rotation matrices, both MRC and wZF benefit from a sharper element pattern, and the gain is particularly pronounced for wZF. By contrast, for random feasible rotations, the wZF performance decreases as $b$ grows. Higher directivity is beneficial only when the boresight vectors are properly coordinated with the propagation geometry. Otherwise, narrower main lobes amplify directional mismatch and can even reduce performance.

\begin{figure}[t]
\centering
\includegraphics[width=2.6in]{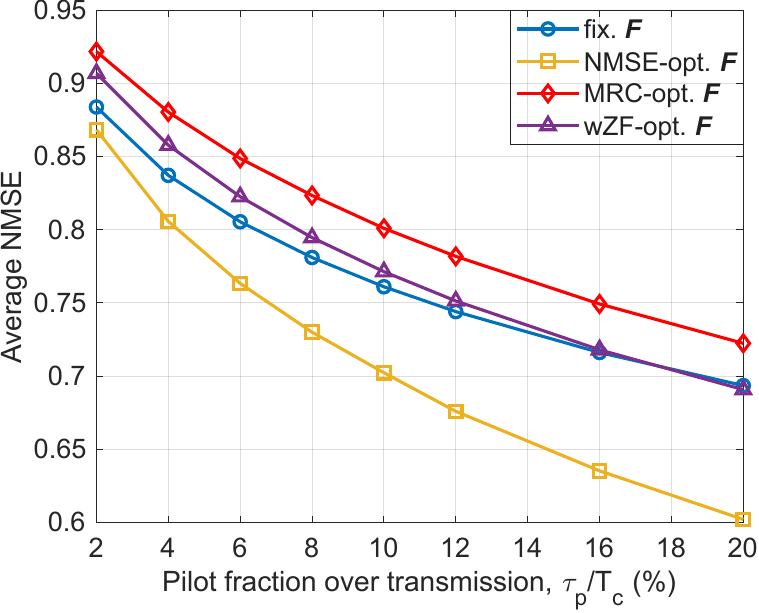}
\caption{Average active-subspace NMSE versus the pilot fraction $\tau_p/T_c$ for different RA-orientation schemes.} 
\label{fig:pilot_nmse}
\vspace{-6pt}
\end{figure}

\begin{figure}[t]
\centering
\includegraphics[width=2.6in]{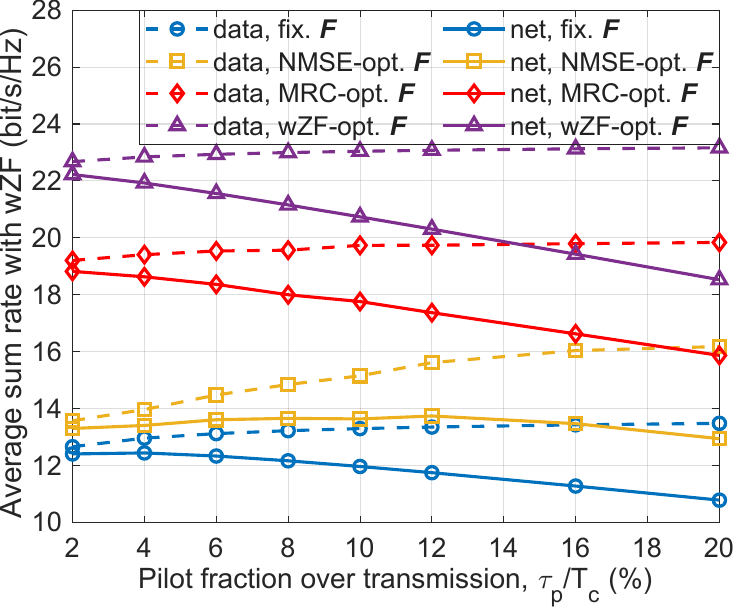}
\caption{Average data-phase and net spectral efficiency versus the pilot fraction under the wZF receiver for different RA-orientation schemes.} 
\label{fig:pilot_rate}
\vspace{-18pt}
\end{figure}

We now turn to the validation of the structural results in Lemma~\ref{lem:NMSE} and Remark~1. Fig.~\ref{fig:pilot_nmse} plots the average active-subspace NMSE versus the pilot fraction $\tau_p/T_c$, while Fig.~\ref{fig:pilot_rate} reports, under the wZF receiver, both the data-phase spectral efficiency $\sum_{k}\log_2(1+\bar{\gamma}^{\mathrm{wZF}}_k)$ and the net spectral efficiency $\eta\sum_k\log_2(1+\bar{\gamma}^{\mathrm{wZF}}_k)$ for four orientation policies: fixed broadside, NMSE-optimal, MRC-optimal, and wZF-optimal. Fig.~\ref{fig:pilot_nmse} confirms the monotonicity predicted by Lemma~\ref{lem:NMSE}: for all rotation policies, increasing the pilot fraction reduces the active-subspace NMSE, and the NMSE-optimal orientation consistently yields the smallest NMSE. Here, the curve labeled “NMSE-opt. $F$” is obtained by solving $\min_{F\in\mathcal{F}} \frac{1}{K}\sum_{k=1}^K \mathrm{NMSE}_k(\boldsymbol{F})$. However, Fig.~\ref{fig:pilot_rate} shows that better estimation quality does not automatically translate into the best communication performance. The wZF-optimal orientation achieves the largest data-phase spectral efficiency throughout the whole range of pilot fractions, while the NMSE-optimal orientation is clearly suboptimal in rate. This is a direct numerical confirmation of Remark~1: minimizing the normalized estimation error and maximizing the communication rate are generally different design objectives because the latter also depends on the deterministic LoS term and, more importantly, on multiuser coupling through the combiner structure. In addition, in Fig.~\ref{fig:pilot_rate}, the dashed curves of data-phase spectral efficiency increase with the pilot fraction because longer training improves channel estimation and, in turn, the effective SINR, while the solid curves of the net spectral efficiency may decrease once the pre-log factor $\eta=1-\tau_p/T_c$ dominates the SINR improvement. In particular, the NMSE-optimal orientation benefits the most from longer training in terms of estimation quality, yet it remains inferior to the wZF-optimal orientation in net rate.

\begin{figure}[t]
\centering
\includegraphics[width=2.6in]{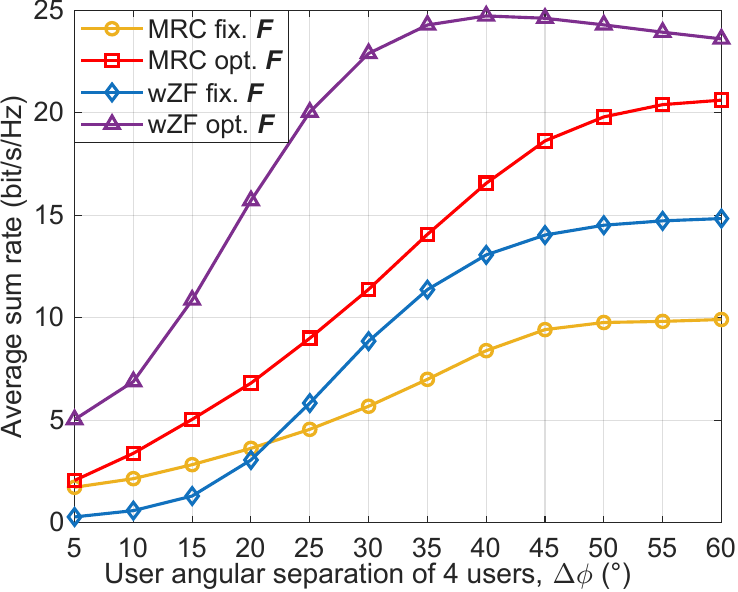}
\caption{Average sum rate versus the user angular separation.} 
\label{fig:ang_sep}
\vspace{-6pt}
\end{figure}

\begin{figure}[t]
\centering
\includegraphics[width=2.6in]{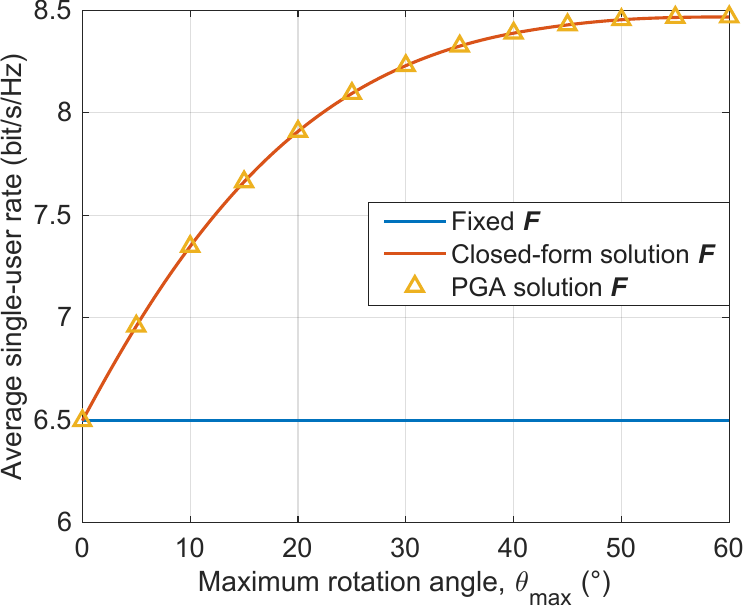}
\caption{Average rate versus the maximum rotation angle of single-user LoS-dominant case based on the closed-form and the projected gradient solution.} 
\label{fig:single}
\vspace{-18pt}
\end{figure}

Fig.~\ref{fig:ang_sep} investigates a controlled four-user geometry in which all users are placed on the same ring while their angular separation $\Delta\phi$ is varied, isolating the geometric mechanism behind the different operating principles of MRC and wZF as depicted in Remark 4 and Corollary~1, respectively.
For MRC, both the fixed and optimized curves increase monotonically and then gradually saturate as $\Delta\phi$ grows. This behavior reflects the fact that MRC primarily benefits from stronger desired channels and reduced interference leakage, but the receiver itself does not actively handle the multiuser geometry. Once the users are sufficiently separated, additional angular spread provides diminishing returns. The optimized MRC curve nevertheless remains substantially above the fixed baseline across the whole range, showing that RA rotation can simultaneously improve the coherent gain term and reduce the multiuser interference terms in \eqref{eq:MRC_I_term}. The wZF behavior is even more informative. At small $\Delta\phi$, the fixed wZF baseline performs poorly because the users are nearly aligned in angle and the effective channels are highly correlated. After RA optimization, however, the wZF rate increases sharply with $\Delta\phi$ and significantly outperforms all MRC curves. This validates the central message of Corollary~\ref{thm:wZF_schur}: RA rotation for wZF is valuable not merely because it strengthens the useful links, but because it reduces the error-aware inter-user coupling while preserving a large error-aware useful gain. Interestingly, the optimized wZF curve reaches its maximum at a moderate-to-large separation and then decreases slightly when $\Delta\phi$ becomes very large. This mild post-peak decline reveals an in-depth RA tradeoff. Once users are already well separated, the cross-correlation penalty in \eqref{eq:wZF_schur_gamma} is nearly minimized, so further angular spreading yields little additional separability gain. At the same time, the finite directional resources of the RA array and the tilt constraint make it harder to maintain uniformly large error-aware self-gains for all users. The result is a transition from a correlation-limited regime to a gain-limited regime, which appears only after the main separability bottleneck has been removed. This figure therefore offers strong evidence that MRC and wZF indeed prefer fundamentally different rotation configurations: MRC is dominated by strength aggregation, while wZF is dominated by error-aware user separability.

Finally, Fig.~\ref{fig:single} validates Proposition~\ref{thm:single_user} in the special single-user LoS-dominant regime. As the maximum rotation angle increases, both the closed-form projection solution in \eqref{eq:single_user_solution} and the PGA-based numerical solution yield a monotonic rate improvement over the fixed broadside baseline, with a diminishing-return profile as $\theta_{\max}$ becomes large. More importantly, the two optimized curves are virtually indistinguishable across the whole range of $\theta_{\max}$. This provides a clean verification of the theory in two senses. First, it confirms that, when the NLoS contribution is negligible, the rotation problem indeed decouples across antennas and reduces to per-element projection onto the spherical cap. Second, it shows that the proposed projected-gradient algorithm incurs no optimality loss in the special regime where a closed-form solution exists. 

\vspace{-6pt}
\section{Conclusion}
\label{sec:conclusion}
This paper investigated uplink multiuser MIMO with a RA array under imperfect CSI. The central observation is that RA rotation is not merely a hardware-side beam-steering refinement. Because the element boresight directions reshape both the Rician channel mean and covariance, rotation affects not only the data-phase effective channels but also the LMMSE estimator and, through it, the residual interference and effective noise seen by the receiver. Motivated by this coupling, we proposed a two-timescale design in which RA orientations are optimized from statistical CSI on a large timescale, while MRC or wZF combiners are updated from instantaneous channel estimates on a small timescale.
Within this framework, we derived a closed-form UatF-based large-timescale rate expression for MRC and a closed-form statistical surrogate for the large-timescale rate for wZF reception. For the resulting non-convex RA rotation optimization problem, we developed a projected-gradient algorithm over a product of spherical caps. Beyond the algorithm itself, the analysis yielded several insights. First, improving normalized estimation quality does not generally imply rate optimality. Second, the single-user LoS-dominant problem admits a closed-form projection solution. Third, MRC and wZF favor fundamentally different rotation behaviors: MRC mainly benefits from signal-strength aggregation, whereas wZF benefits from stronger error-aware user separation.
The numerical results validated both the accuracy of the proposed surrogates and the practical gains of the two-timescale RA design. In particular, the largest benefits were observed precisely in the regimes where conventional systems struggle most when multiuser geometry is unfavorable. Future work may extend this framework to wideband systems, downlink precoding, and joint position-orientation adaptive arrays.

\end{document}